\documentclass{scrartcl}
\def\arXiv{}

\usepackage{basicstuff}

\usepackage{tikz}

\usepackage{pifont}
\usepackage{booktabs}

\newcommand{\cmark}{\ding{51}}%
\newcommand{\xmark}{\ding{55}}%

\newcommand{\yes}{\cmark}
\newcommand{\no}{\xmark}

\renewcommand{\todo}[1]{}

\newcommand{\ii}[1]{{\normalfont \ensuremath{#1^{\tiny\circled{i}}}}}%
\renewcommand{\k}[1]{{\normalfont \ensuremath{#1^{\tiny\circled{k}}}}}%
\newcommand{\jj}[1]{{\normalfont \ensuremath{#1^{\tiny\circled{j}}}}}%

\renewcommand{\ii}[1]{{\normalfont \ensuremath{#1_{i}}}}%
\renewcommand{\k}[1]{{\normalfont \ensuremath{#1_{k}}}}%
\renewcommand{\jj}[1]{{\normalfont \ensuremath{#1_{j}}}}%

\newenvironment{sketch}{\par\vspace{1pt} %
  \noindent{\textit{Proof Sketch.}} %
  \hspace*{0em}} %
{\nopagebreak\hfill$\Box$\par\vspace{6pt}}

\ifdefined\arXiv

\usepackage{arXiv}

\title{\Large Simultaneous Embedding of Planar Graphs\thanks{Submitted as a
    chapter about simultaneous embedding to the GD Handbook edited by
    Roberto Tamassia.}}

\author{Thomas Bläsius$^{1)}$, Stephen G.~Kobourov$^{2)}$, Ignaz
  Rutter$^{1)}$}

\date{$^{1)}$ Karlsruhe Institute of Technology (KIT) \\ %
  \texttt{\{blaesius,rutter\}@kit.edu}\\
  \medskip
  $^{2)}$ University of Arizona \\ %
  \texttt{kobourov@cs.arizona.edu}}

\fi

\begin{document}

\maketitle


\begin{abstract}
  Simultaneous embedding is concerned with simultaneously representing
  a series of graphs sharing some or all vertices.  This forms the
  basis for the visualization of dynamic graphs and thus is an
  important field of research.  Recently there has been a great deal
  of work investigating simultaneous embedding problems both from a
  theoretical and a practical point of view.  We survey recent work on
  this topic.
\end{abstract}

\section{Introduction}
\label{sec:introduction}

Traditional problems in graph drawing involve the layout of a single
graph, whereas in simultaneous graph drawing we are concerned with the
layout of multiple related graphs. In particular, consider the problem
of drawing a series of graphs that share all, or parts of the same
vertex set. The graphs may represent different relations between the
same set of objects, or alternatively, the graphs may be the result of
a single relation that changes through time.  

In this chapter we survey efforts to address the following problem:
Given a series of graphs that share all, or parts of the same vertex
set, what is a natural way to layout and display them? The layout and
display of the graphs are different aspects of the problem, but also
closely related, as a particular layout algorithm is likely to be
matched best with a specific visualization technique. As stated above,
however, the problem is too general and it is unlikely that one
particular layout algorithm will be best for all possible
scenarios. Consider the case where we only have a pair of graphs in
the series, and the case where we have hundreds of related graphs. The
``best'' way to layout and display the two series is likely going to
be different. Similarly, if the graphs in the sequence are very
closely related or not related at all, different layout and display
techniques may be more appropriate.

For the layout of the graphs, there are two important criteria to
consider: the {\em readability} of the individual layouts and the {\em
  mental map preservation} in the series of drawings.  The readability
of individual drawings depends on aesthetic criteria such as display
of symmetries, uniform edge lengths, and minimal number of
crossings. Preservation of the mental map can be achieved by ensuring
that vertices that appear in consecutive graphs in the series, remain
in the same positions. These two criteria are often contradictory. If
we individually layout each graph, without regard to other graphs in
the series, we may optimize readability at the expense of mental map
preservation. Conversely, if we fix the vertex positions in all
graphs, we are optimizing the mental map preservation but the
individual layouts may be far from readable. In simultaneous graph
embedding, vertices are placed in the exact same locations in all the
graphs, while the layout of the edges may differ.

Visualization of related graphs, that is, graphs that are defined on
the same set of vertices, arise in many different settings. Software
engineering, databases, and social network analysis, are all examples
of areas where multiple relationships on the same set of objects are
often studied. In evolutionary biology, phylogenetic trees are used to
visualize the ancestral relationship among groups of
species. Depending on the assumptions made, different algorithms
produce different phylogenetic trees.  Comparing the outputs and
determining the most likely evolutionary hypothesis can be difficult
if the drawings of the trees are laid out independently of each other.

While in some of the above examples the graphs are not necessarily
planar, solving the planar case can provide intuition and ideas for
the more general case. With this in mind, here we concentrate on the
problem of simultaneous embedding of planar graphs.  Simultaneous
embedding of planar graphs generalizes the notion of traditional graph
planarity and is motivated by its relationship with problems of graph
thickness, geometric thickness, and applications such as the
visualization of graphs that evolve through time.

The thickness of a graph is the minimum number of planar subgraphs
into which the edges of the graph can be partitioned;
see~\cite{mutzelthickness98} for a survey. Thickness is an important
concept in VLSI design, since a graph of thickness $k$ can be embedded
in $k$ {\em layers}, with any two edges drawn in the same layer
intersecting only at a common vertex and vertices placed in the same
location in all layers.  A related graph property is {\em geometric
  thickness}, defined to be the minimum number of layers for which a
drawing of $G$ exists having all edges drawn as straight-line
segments~\cite{deh-gtcg}. Finally, the {\em book thickness} of a graph
$G$ is the minimum number of layers for which a drawing of $G$ exists,
in which edges are drawn as straight-line segments and vertices are in
convex position~\cite{bk-btg-79}. It has been shown that the book
thickness of planar graphs is no greater than
four~\cite{JCSS::Yannakakis1989}.

\subsection{Problem Definitions}

This chapter is structured along three basic simultaneous embedding
results for planar graphs, {\sc Simultaneous Geometric Embedding}
({\sc SGE}), {\sc Simultaneous Embedding with Fixed Edges} ({\sc
  SEFE}), and {\sc Simultaneous Embedding} ({\sc SE}), \todo{new
  figure}Figure~\ref{fig:sge-se-sefe} illustrates the three cases.
For all three problems the input always consists of two planar graphs
$\1G = (\1V, \1E)$ and $\2G = (\2V, \2E)$ sharing a common subgraph $G
= (V, E) = (\1V \cap \2V, \1E \cap \2E)$.  

The most strict variant is {\sc Simultaneous Geometric Embedding}
({\sc SGE}), which asks for planar straight-line drawings of~$\1G$
and~$\2G$ such that common vertices have the same coordinates in both
drawings.  The requirements of SGE are very strict, and as we will see
in Section~\ref{sec:sge} there exist a lot of examples that do not
admit such an embedding.  While the problem {\sc Simultaneous
  Embedding with Fixed Edges} still requires common vertices to have
the same coordinates, it relaxes the straight-line requirement by
allowing arbitrary curves for representing edges.  To maintain the
mental map, common edges are still required to be represented by the
same curves.  Finally, {\sc Simultaneous Embedding} drops the
constraints on the curves altogether and just requires common vertices
to have the same coordinates.

For all these problems it is common to also use the problem name to
denote a corresponding embedding, that is we also say that $\1G$ and
$\2G$ have an {\sc SGE}, {\sc SEFE} or {\sc SE} if they admit
solutions to these problems.  Moreover, all these problems readily
generalize to~$k > 2$ input graphs~$\1G,\dots,\k G$, by requiring that
the conditions hold for each pair of graphs.  In this case a common
restriction is to require that all input graphs share exactly the same
graph~$G$, that is~$G = \ii G \cap \jj G$ for~$i \ne j$.  We call this
behavior \emph{sunflower intersection}.

\begin{figure}
  \centering
  \includegraphics{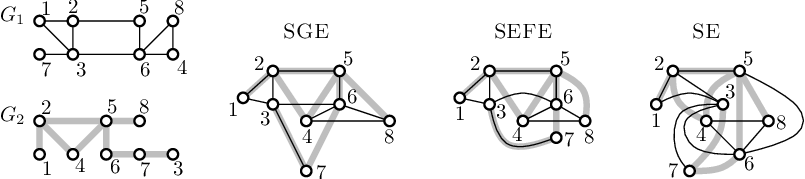}
  \caption{Two graphs $\1G$ and $\2G$ together with an {\sc SGE}, a
    {\sc SEFE} and an {\sc SE}.  In the {\sc SGE} all edges are
    straight line segments while some edges in the {\sc SEFE} are not.
    The {\sc SE} contains common edges ($\{3, 7\}$ and $\{5, 6\}$)
    that are drawn differently with respect to $\1G$ and~$\2G$.  
    \label{fig:sge-se-sefe}}
\end{figure}

We note that simultaneous embedding problems are closely related to
constrained embedding problems.  For example if the planar embedding
of one of the two graphs of an instance of {\sc SEFE} is already
fixed, the problem of finding a {\sc SEFE} is equivalent to finding an
embedding of the second graph respecting a prescribed embedding for a
subgraph, namely the common graph.  This constrained embedding problem
is known as {\sc Partially Embedded Planarity}.  Angelini et
al.~\cite{abf-10} show that this problem can be solved in linear time
and, in the spirit of Kuratowski's theorem, Jel\'inek et
al.~\cite{jkr-kttppeg-11} characterize the yes-instances by forbidden
substructures.  A similar tie to constrained embedding problems exists
in the case of {\sc SE}.  After fixing the drawing of one of the two
input graphs it remains to draw a single graph without crossings at
prescribed vertex positions.  This problem is known as {\sc Point Set
  Embedding} and Pach and Wenger show that this is always
possible~\cite{pw-epgfvl-98}.  There are other, less obvious relations
between simultaneous embedding and constrained embedding problems,
which will be described later.

\subsection{Overview and Outline}

This chapter starts with the three simultaneous embedding problems
{\sc SGE}, {\sc SEFE}, and {\sc SE}, and we discuss each of them in
one of the following sections.  There are three major classes of
results on simultaneous embedding problems.  The first class contains
algorithms that, for given graphs with certain properties, always
produce a simultaneous embedding, perhaps with additional quality
guarantees.  These results show the existence of simultaneous
embeddings for the corresponding graph classes.  The second class
contains counterexamples that do not admit a simultaneous embedding.
The third class contains algorithms and complexity results for the
problem of testing whether a given instance admits a simultaneous
embedding.

We present a survey of the results on {\sc SGE} in
Section~\ref{sec:sge}.  Due to the strong requirements of {\sc SGE}
results of the first type, which identify classes of graphs that
always admit a simultaneous embedding, exist only for very few and
strongly restricted graph classes.  For example, even a path and a
tree of depth~4 may not have an {\sc SGE}~\cite{kaufmann-tp-12}.
Moreover, it is NP-hard to decide {\sc SGE} and there are no further
results of the third type, that is algorithms testing whether an
instance has an {\sc SGE} or not, even for restricted instances.

Section~\ref{sec:sefe} presents the {\sc SEFE} problem, which turns
out to be much less restrictive than {\sc SGE}.  For example a tree
and a path do always admit a {\sc SEFE} although they do not have an
{\sc SGE}~\cite{fab-sefe-06}.  On the other hand, examples not having
a {\sc SEFE} are also counterexamples for {\sc SGE}.  Moreover, it is
still open whether {\sc SEFE} can be tested in polynomial time for two
graphs, whereas it is NP-complete for three or more
graphs~\cite{juenger-sefe-06}.  However, for two graphs, there exist
several results of the third type, that is testing algorithms, for
restricted inputs.  For example, it is possible to decide in linear time
whether a pair of graphs admits a {\sc SEFE} or not, if the common
graph is biconnected~\cite{adfpr-tsegi-12, lubiw-testing-10}.

In Section~\ref{sec:se}, we consider the least restrictive
simultaneous embedding problem, {\sc SE}, which only requires common
vertices to have the same coordinates in all drawings.  As every
planar graph can be drawn without crossings even if the position of
every vertex is fixed~\cite{pw-epgfvl-98}, there are no
counterexamples for {\sc SE} and it is not necessary to have a testing
algorithm.  The results on {\sc SE} focus on creating simultaneous
embeddings such that edges have few bends and the resulting drawings
use small area.

Sections~\ref{sec:color-se}--\ref{sec:practical-approaches} presents
several variants of approaches to simultaneous embedding that do not
quite fall into the categories of the three main problems.  The
problem variants discussed in Section~\ref{sec:color-se} relax the
requirement of having a fixed mapping between the vertices of~$\1G$
and~$\2G$.  They rather ask whether a suitable mapping can be found
such that a {\sc SEFE} exists~\cite{se-original-07}.  Colored {\sc
  SGEs} are somewhere between and allow the mapping to identify only
vertices having the same color~\cite{colored-se-11}.
Section~\ref{sec:matched-drawings} deals with matched drawings
requiring straight-line drawings of the two input graph such that each
common vertex has only the same $y$-coordinate in both drawings.
Other work, discussed in Section~\ref{sec:other-simult-repr}, deals
with the problem of simultaneously representing a planar graph and its
dual~\cite{t-hdg-63} and considers different types of simultaneous
representations, such as simultaneous intersection representations, as
introduced by Jampani and Lubiw~\cite{lubiw-chordal-09}.
Section~\ref{sec:practical-approaches} presents several practical
approaches to simultaneous embedding problems.

\todo{new paragraph added}In Section~\ref{sec:morphing} results on
morphing between different planar drawings of the same graph are
presented.  A morph aims to preserve the mental map between different
drawings of the same graph, which can be seen as the opposite to
drawing different graphs such that the common part is drawn the same.
Finally, in Section~\ref{sec:open-problems}, we present a list of open
questions.  The list
contains questions that have been open for several years, as
well as questions that are motivated by recent research results.

\section{Simultaneous Geometric Embedding}
\label{sec:sge}

In this section we consider the most desirable (and most restrictive)
kind of simultaneous drawings, the {\sc SGE}s.  Most results on that
problem are summarized in \todo{table added}Table~\ref{tab:sge}.
Figure~\ref{fig:sge} illustrates the relation between these results.
Before we describe the results in more detail we start with a small
example.  While it may be tempting to say that if the union of two
graphs contains a subdivision of $K_5$ or $K_{3,3}$ then the two
graphs have no simultaneous geometric embedding, this is not the case;
see Figure~\ref{fig:union}. In fact, while planarity testing for a
single graph can be done in linear time~\cite{ht-ept-74},
Estrella-Balderrama et al.~\cite{juenger-np-07} show that the decision
problem {\sc SGE} is NP-hard.  Other results concerning the complexity
of {\sc SGE} (for example for restricted graph classes) are not known.

In the following we describe the results illustrated in
Figure~\ref{fig:sge}.  We start with algorithms always creating an
{\sc SGE} when the input is restricted to special graph
classes.  We then continue with graph classes containing
counterexamples.  Finally, we consider the results not fitting in one
of these two cases.

\begin{table}[tb]
  \newcommand{\row}{\rule{0pt}{2.5ex}}
  \centering
  \begin{tabular}{lccc}
    \toprule 
    \row 
    {\bf SGE Instance}                           & {\bf Existence} & {\bf Area}                             & {\bf Ref.}                           \\
    \midrule
    $\1G$ \& $\2G$ paths                         & \yes            & $n \times n$                           & {\cite{se-original-07}}              \\
    $\1G$ path \& $\2G$ extended star            & \yes            & $O(n^2) \times O(n)$                   & {\cite{se-original-07}}              \\
    $\1G$ caterpillar \& $\2G$  path             & \yes            & $n \times 2n$                          & {\cite{se-original-07}}              \\
    $\1G$ \& $\2G$ caterpillar                   & \yes            & $3n \times 3n$                         & {\cite{se-original-07}}              \\
    $2$ stars                                    & \yes            & $3 \times (n-2)$                       & {\cite{se-original-07}}              \\
    $k$ stars                                    & \yes            & $O(c^k\sqrt{n}) \times O(c^k\sqrt{n})$ & {\cite{se-original-07}}              \\
    $\1G$ \& $\2G$ cycles                        & \yes            & $4n \times 4n$                         & {\cite{se-original-07}}              \\
    $\1G$ \& $\2G$ have maximum degree 2         & \yes            & ---                                    & {\cite{dek-gtldg-04}}                \\
    $\1G$ wheel \& $\2G$ cycle                   & \yes            & ---                                    & {\cite{beppe-matched-11}}            \\
    $\1G$ tree \& $\2G$ matching                 & \yes            & ---                                    & {\cite{beppe-matched-11}}            \\
    $\1G$ outerpath \& $\2G$ matching            & \yes            & ---                                    & {\cite{beppe-matched-11}}            \\
    $\1G$ tree of depth 2 \& $\2G$ path          & \yes            & ---                                    & {\cite{kaufmann-tp-12}}              \\
    $\1G$ level-planar w.r.t. path $\2G$         & \yes            & ---                                    & {\cite{se-arcs-09}}                  \\
    \midrule
    $\1G$ \& $\2G$ planar                        & \no             & ---                                    & {\cite{se-original-07}}              \\
    $\1G$ path \& $\2G$ planar                   & \no             & ---                                    & {\cite{se-original-07,ek-sepgfb-05}} \\
    $\1G$ path \& $\2G$ edge disjoint            & \no             & ---                                    & {\cite{se-constrained-09}}           \\
    three paths                                  & \no             & ---                                    & {\cite{se-original-07}}              \\
    $\1G$ matching \& $\2G$ planar               & \no             & ---                                    & {\cite{beppe-matched-11}}            \\
    six matchings                                & \no             & ---                                    & {\cite{beppe-matched-11}}            \\
    $\1G$ \& $\2G$ outerplanar                   & \no             & ---                                    & {\cite{se-original-07}}              \\
    $\1G$ \& $\2G$ trees                         & \no             & ---                                    & {\cite{kaufmann-trees-09}}           \\
    $\1G$ depth-4 tree \& $\2G$ edge disj.\ path & \no             & ---                                    & {\cite{kaufmann-tp-12}}              \\
    \bottomrule
  \end{tabular}
  \caption{A list of classes of graphs that are either known to always have
    an {\sc SGE}  or that contain counterexamples.  For the positive cases, the area consumption is given, provided that it is known.}
  \label{tab:sge}
\end{table}

\begin{figure}[hbt]
  \centering
  \includegraphics{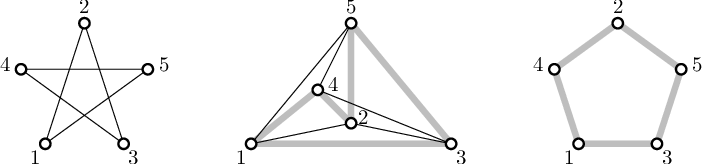}
  \caption{The union of the graph on the left and the graph on the
    right is a $K_5$, but the middle drawing shows a simultaneous
    geometric embedding of the two graphs.
  \label{fig:union}
}
\end{figure}

\begin{figure}[t]
  \centering
    \begin{tikzpicture}[overlay]
      \node at (-7.75*16pt, -1pt-0*16pt) [anchor=north]
      {\cite{se-original-07}}; 
      \node at (0*16pt, -1pt-0*16pt) [anchor=north]
      {\cite{se-original-07}}; 
      \node at (8*16pt, -1pt-0*16pt) [anchor=north]
      {\cite{se-original-07}}; 
      \node at (-7.75*16pt, -1pt-4*16pt) [anchor=north]
      {\cite{se-original-07}}; 
      \node at (-0.75*16pt, -1pt-4*16pt) [anchor=north]
      {\cite{se-original-07}}; 
      \node at (5*16pt, -1pt-4*16pt) [anchor=north]
      {\cite{se-original-07}}; 
      \node at (9.25*16pt, -1pt-4*16pt) [anchor=north]
      {\cite{se-original-07}}; 
      \node at (-7*16pt, -1pt-8*16pt) [anchor=north]
      {\cite{beppe-matched-11}}; 
      \node at (0*16pt, -1pt-8*16pt) [anchor=north]
      {\cite{dek-gtldg-04}}; 
      \node at (-5*16pt, -1pt-11*16pt) [anchor=north]
      {\cite{beppe-matched-11}}; 
      \node at (-6*16pt, -1pt-14*16pt) [anchor=north]
      {\cite{beppe-matched-11}}; 
      \node at (6.5*16pt, -1pt-11*16pt) [anchor=north]
      {\cite{kaufmann-tp-12}}; 
      \node at (6*16pt, -1pt-14*16pt) [anchor=north]
      {\cite{se-arcs-09}}; 
      \node at (-6*16pt, -1pt-18*16pt) [anchor=north]
      {\cite{se-constrained-09}}; 
      \node at (5*16pt, -1pt-18*16pt) [anchor=north]
      {\cite{kaufmann-tp-12}}; 
      \node at (-7*16pt, -1pt-21*16pt) [anchor=north]
      {\cite{se-original-07,ek-sepgfb-05}}; 
      \node at (-1.5*16pt, -1pt-21*16pt) [anchor=north]
      {\cite{beppe-matched-11}}; 
      \node at (3.5*16pt, -1pt-21*16pt) [anchor=north]
      {\cite{kaufmann-trees-09}}; 
      \node at (8.5*16pt, -1pt-21*16pt) [anchor=north]
      {\cite{beppe-matched-11}}; 
      \node at (-7*16pt, -1pt-24*16pt) [anchor=north]
      {\cite{se-original-07}}; 
      \node at (2.5*16pt, -1pt-24*16pt) [anchor=north]
      {\cite{se-original-07}}; 
      \node at (8.5*16pt, -1pt-24*16pt) [anchor=north]
      {\cite{se-original-07}}; 
    \end{tikzpicture}

  \includegraphics{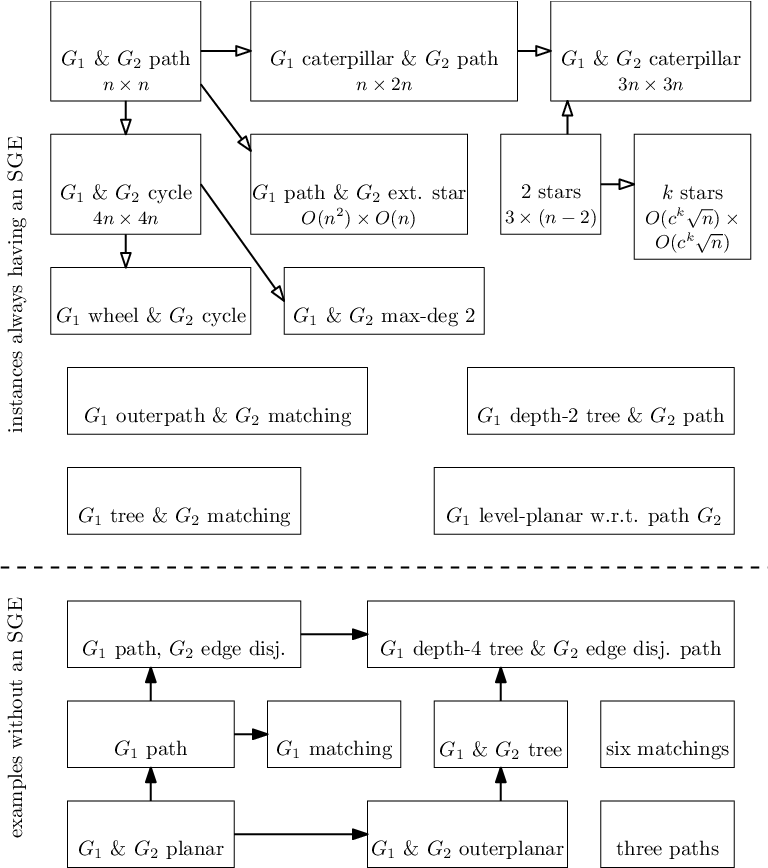}
  \caption{Overview over the so far known results on {\sc SGE}.  Each
    box represents one result and an arrow highlights that the
    source-result is extended by the target-result.  The arrowheads
    are empty for the cases in which this is only true if the grid
    size is neglected.  Note that transitive arrows are omitted.
  \label{fig:sge}
}

\bigskip
\end{figure}

\subsection{Graph Classes with SGE}

Brass et al.~\cite{se-original-07} give several algorithms for
different restricted graph classes always creating an {\sc SGE}.  In
the simplest case $\1G$ and $\2G$ are both required to be paths.  This
result is easy to prove and also provides good intuition for most of
the positive results: 

\begin{theorem}
  \label{the-two-path-theorem}
  For two paths $\1P$ and $\2P$ on the same vertex set $V$ of size $n$
  an {\sc SGE} on a grid of size $n \times n$ can be found in linear
  time.
\end{theorem}
\begin{proof}
  For each vertex $u\in V$, we embed $u$ at the integer grid point
  $(\1p, \2p)$, where $\ii p\in\{1,2,\ldots,n\}$ is the vertex's
  position in the path $\ii P$, $i \in \{1,2\}$.  Then, $\1P$ is
  embedded as an $x$-monotone polygonal chain, and $\2P$ is embedded
  as a $y$-monotone chain. Thus, neither path is self-intersecting;
  see Figure~\ref{fig:twopaths} for an example.
\end{proof}

\begin{figure}[t]
  \centering
  \includegraphics{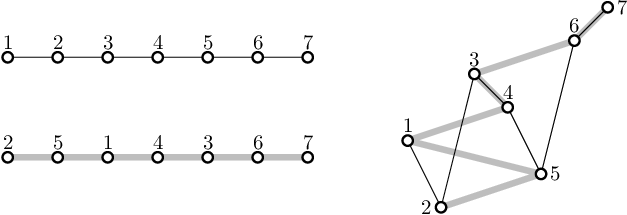}
  \caption{Two paths simultaneously embedded such that one path is
    $x$-monotone and the other is $y$-monotone.
  \label{fig:twopaths}
}
\end{figure}

Brass et al.~\cite{se-original-07} also consider more general graph
classes, such as \emph{caterpillars} (trees being paths after the
removal of all leaves), \emph{stars} (trees with at most one inner
vertex called \emph{center}), and \emph{extended stars} (collection of
stars with an additional special root and paths from the special root
to the centers of all stars).  They show that a caterpillar and a path
admit an {\sc SGE} on a grid of size $n \times 2n$, which can be
extended to two caterpillars on a grid of size $3n \times 3n$.
Moreover, they can simultaneously embed two stars on a $3 \times
(n-2)$ grid and extend it to the case of $k$ stars on an
$O(c^k\sqrt{n}) \times O(c^k\sqrt{n})$-grid, where $c$ is a constant.
Finally, the pairs path plus extended star and cycle plus
cycle can be embedded on $O(n^2) \times O(n)$ and $4n \times 4n$
grids, respectively.  The latter two results both extend the case of
two paths (when neglecting the grid size).

The result for two cycles was further extended by Duncan et
al.~\cite{dek-gtldg-04} and Cabello et al.~\cite{beppe-matched-11}.
Duncan et al.~\cite{dek-gtldg-04} show that a graph with maximum
degree~4 has geometric thickness~2.  To this end, they show that two
graphs with maximum degree~2 always admit a simultaneous geometric
embedding.  However, their algorithm computes drawings with
potentially large area.  

Cabello et al.~\cite{beppe-matched-11} show the existence of an {\sc
  SGE} for a \emph{wheel} (union of a star and a cycle on its leaves)
and a cycle.  They moreover give algorithms for the pairs tree plus
\emph{matching} (graph with maximum degree~1) and \emph{outerpath}
(outerplanar graph whose weak dual is a path) plus matching.  The
former algorithm uses only two slope for the matching edges, for the
latter one slope suffices.

Given a planar graph and a path on the same vertices, the order of the
vertices in the path induces a layering on the vertices.  Cappos et
al.~\cite{se-arcs-09} give a linear-time algorithm that computes an
{\sc SGE} of a planar graph and a path if the planar graph is
level-planar with respect to the layering induced by the path.
Angelini et al.~\cite{kaufmann-tp-12} show that every tree of depth~2
has an {\sc SGE} with every path.

\subsection{Examples Without SGE}

In contrast to the positive results, Brass et
al.~\cite{se-original-07} give several examples not admitting an {\sc
  SGE}.  They show the existence of two planar graph without a
simultaneous embedding and extended this result to two outerplanar
graphs.  Two results we present in more detail are the counterexample
for a planar graph and a path by Brass et al.~\cite{se-original-07}
and Erten and Kobourov~\cite{ek-sepgfb-05} and the counterexample of
three paths by Brass et al.~\cite{se-original-07}.

\begin{theorem}
  \label{counter-example}
  There exists a planar graph $G$ and a path $P$ not admitting an {\sc
    SGE}.
\end{theorem}
\begin{sketch}
  Consider the graph $G$ and the path $P$ as shown in
  Figure~\ref{fig:pathplanar}.  Let $G'$ be the subgraph of $G$
  induced on the vertices $\{1,2,3,4,5\}$, and let $G''$ be the
  subgraph of $G$ induced on the vertices $\{2,6,7,8,9\}$.  Since $G$
  is triconnected fixing the outer face fixes an embedding for
  $G$. With the given outer face of $G$, the path $P$ contains two
  crossings: one involving $(2,4)$, and the other one involving
  $(6,8)$.  

  Graph $G'$ has six faces and unless we change the outer
  face of $G'$ such that it contains the edge $(1,3)$ or $(3,5)$, the
  edge $(2,4)$ is involved in a crossing in the path. Similarly for
  $G''$, unless we change its outer face such that it contains $(2,7)$
  or $(7,9)$, the edge $(6,8)$ is involved in a crossing in the
  path. However $G'$ and $G''$ do not share any faces and removing
  both crossings depends on taking two different outer faces, which is
  impossible. Thus, regardless of the choice for the outer face of
  $G$, path $P$ contains a crossing.
\end{sketch}

\begin{figure}[t]
  \centering
  \includegraphics{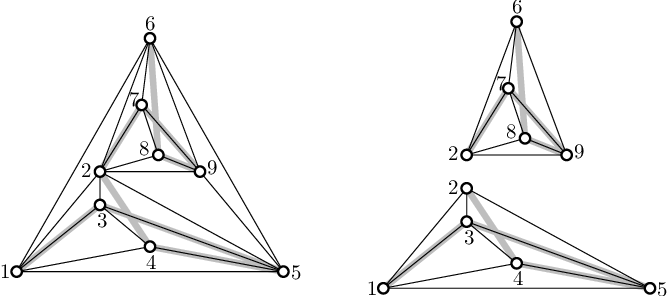}
  \caption{A planar graph $G$ and a path $P$ that do not allow an {\sc
      SGE}.
\label{fig:pathplanar}
}

\bigskip
\end{figure}

\begin{theorem}
  \label{the-three-path-theorem}
  There exist three paths $\1P$, $\2P$ and $\3P$ not admitting an {\sc
    SGE}.
\end{theorem}
\begin{proof}
  A path of $n$ vertices is simply an ordered sequence of $n$ numbers.
  The three paths we consider are: 714269358, 824357169 and 758261439.
  For example, the sequence $714269358$ represents the path
  $(v_7,v_1,v_4,v_2,v_6,v_9,v_3,v_5,v_8)$.  We will write $ij$ for the
  edge connecting $v_i$ to $v_j$.  The union of these paths contain
  the following twelve edges.
  
  $$E=\{14, 16, 17, 24, 26, 28, 34, 35, 39, 57, 58, 69\}$$
 
  It is easy to see that the graph $G$ consisting of these edges is a
  subdivision of $K_{3,3}$ and therefore non-planar: collapsing 1 and
  7, 2 and 8, 3 and 9 yields the classes $\{1,2,3\}$ and~$\{4,5,6\}$.

  It follows that there are two nonadjacent edges of $G$ that cross
  each other. It is easy to check that every pair of nonadjacent edges
  from $E$ appears in at least one of the paths given above.
  Therefore, at least one path will cross itself which completes the
  proof.
\end{proof}

Cabello et al.~\cite{beppe-matched-11} extend the counterexample for
the case that $\1G$ is a path to the case where $\1G$ is a matching.
Moreover, they give an example of six matchings not admitting an {\sc
  SGE}.  Note that this does not directly follow by dividing three
paths without an {\sc SGE} into six matchings, as the resulting
matchings allow crossings that were not allowed before.  Another
extension of the case where $\1G$ is a path was given by Frati et
al.~\cite{se-constrained-09} who give a counterexample where $\1G$ is
a path and $G$ is a set of isolated vertices, that is $\1G$ and $\2G$
are edge disjoint.

The question of whether two trees always admit an {\sc SGE} was open
for several years, before it was answered in the negative by Geyer et
al.~\cite{kaufmann-trees-09} with a construction involving two very
large trees.  This of course extends the result of two outerplanar
graphs not having an {\sc SGE} by Brass et al.~\cite{se-original-07}.
Angelini et al.~\cite{kaufmann-tp-12} further extended it to the case
of a tree and a path without an {\sc SGE}.  More precisely, they give
an example of a tree of depth~4 and an edge disjoint path not having
an {\sc SGE}.  Recall that a tree of depth~2 does always admit a
simultaneous embedding with a path, thus in this case the gap between
positive and negative results is quite small.

\subsection{Related Work}

Frati et al.~\cite{se-constrained-09} consider the restricted case
where each input graph has a prescribed combinatorial embedding.  They
show that the pair path plus star admits an {\sc SGE} even if the
embedding of the star is fixed.  They can extend this result to a
\emph{double-star} (tree with up to two inner vertices) if it is edge
disjoint to the path.  On the other hand they show that fixing the
embedding of two caterpillars may lead to an counterexample, whereas
they admit an {\sc SGE} if the embedding is not fixed.  Another
counterexample is the pair outerplanar graph with fixed embedding plus
edge-disjoint path.

An interesting additional restriction to {\sc SGE}s was considered by
Argyriou et al.~\cite{abks-grsdg-12}, combining {\sc SGE} with the RAC
drawing convention (RAC -- Right-Angular Crossing).  They try to find
an {\sc SGE} such that crossings between exclusive edges of different
graphs are restricted to right-angular crossings.  Argyriou et
al. consider only the case where the edge sets of both graphs are
disjoint.  They present one negative and one positive result for this
problem.  The negative result consists of a wheel and a cycle not
admitting an {\sc SGE} with right-angular crossings.  On the other
hand they show the existence of such a drawing on a small integer grid
for the case that one of the graphs is a path or a cycle and the other
is a matching.  Moreover, they give a linear-time algorithm to compute
such a drawing.

\section{Simultaneous Embedding with Fixed Edges}
\label{sec:sefe}

\todo{new introduction and new figure}In this section we drop the
requirement that edges have to be straight line segments and consider
the {\sc SEFE} problem.  Figure~\ref{fig:sefe-example} shows a {\sc
  SEFE} of the graph and the path from Figure~\ref{fig:pathplanar} not
admitting an {\sc SGE}.  Figure~\ref{fig:sefe-overview} and
\todo{table added}Table~\ref{tab:sefe} illustrate the results on the
problem {\sc SEFE} classified in the three categories described
before.

\begin{figure}[hbt]
  \centering
  \includegraphics{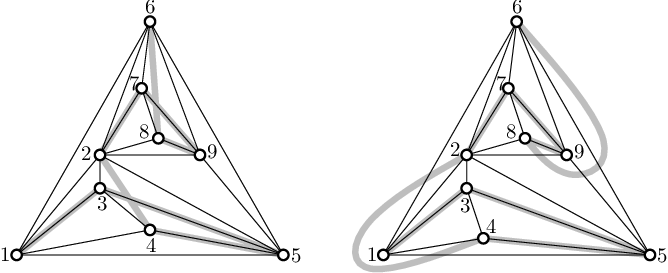}
  \caption{A graph and a path not admitting an {\sc SGE} but a {\sc
      SEFE}.
  \label{fig:sefe-example}}
\end{figure}

\begin{table}[p]
  \heavyrulewidth=1pt
  \newcommand{\row}{\rule{0pt}{2.5ex}}
  \centering
  \begin{tabular}{lcccc} 
    \toprule
    \row
    {\bf SEFE Instance}                                                                                     & {\bf Exist.}                          & {\bf Area}                & {\bf Bends} & {\bf Ref.}              \\
    \midrule
    $\1G$ tree \& $\2G$ path                                                                                & \yes                                  & $O(n) \times O(n^2)$      & 1 \& 0      & {\cite{ek-sepgfb-05}}   \\
    $\1G$ outerplanar \& $\2G$ path                                                                         & \yes                                  & $O(n) \times O(n^2)$      & 1 \& 0      & {\cite{beppe-outer-07}} \\
    $\1G$ outerplanar \& $\2G$ cycle                                                                        & \yes                                  & $O(n^2) \times O(n^2)$    & 1           & {\cite{beppe-outer-07}} \\
    $\1G$, $\2G$ outerpl. \& $G$ collection of paths                                                        & \yes                                  & $O(n^2) \times O(n^2)$    & 1           & {\cite{beppe-outer-07}} \\
    $\1G$ tree \& $\2G$ planar                                                                              & \yes                                  & \multicolumn{2}{c}{ --- } & {\cite{fab-sefe-06}}                  \\
    $\1G$ pseudoforest, $\2G$ planar \& $G$ forest                                                          & \yes                                  & \multicolumn{2}{c}{ --- } & {\cite{fab-sefe-06}}                  \\
    $\1G$ has disj. cycles, $\2G$ planar \& $G$ forest                                                      & \yes                                  & \multicolumn{2}{c}{ --- } & {\cite{juenger-sefespqr-09}}          \\
    \midrule
    characterization of $G$                                                                                 & \yes / \no                            & \multicolumn{2}{c}{ --- } & {\cite{juenger-sefeinter-09}}         \\
    characterization of $\1G$                                                                               & \yes / \no                            & \multicolumn{2}{c}{ --- } & {\cite{sefe-11}}                      \\
    characterization of $\1G$ ($\1G$, $\2G$ outerpl.)                                                       & \yes / \no                            & \multicolumn{2}{c}{ --- } & {\cite{sefe-11}}                      \\
    \midrule
    $\1G$ outerplanar \& $\2G$ planar                                                                       & \no                                   & \multicolumn{2}{c}{ --- } & {\cite{se-original-07}}               \\
    $k$ outerplanar graphs                                                                                  & \no                                   & \multicolumn{2}{c}{ --- } & {\cite{se-original-07}}               \\
    three paths                                                                                             & \no                                   & \multicolumn{2}{c}{ --- } & {\cite{se-original-07}}               \\
    $\1G$ \& $\2G$ outerplanar                                                                              & \no                                   & \multicolumn{2}{c}{ --- } & {\cite{fab-sefe-06}}                  \\
     \bottomrule
     \toprule
    \multicolumn{2}{l}{\row{\bf SEFE Instance}}                                                             & \multicolumn{2}{c}{{\bf Complexity} } & {\bf Ref.}                                                        \\
    \midrule
    \multicolumn{2}{l}{\row three planar graphs}                                                            & \multicolumn{2}{c}{NP-complete }      & {\cite{juenger-sefe-06}}                                          \\
    \multicolumn{2}{l}{\row $\1G$ pseudoforest \& $\2G$ planar}                                             & \multicolumn{2}{c}{$O(n)$ }           & {\cite{juenger-sefespqr-09}}                                      \\
    \multicolumn{2}{l}{\row $\1G$ has $\le 2$ cycles, $\2G$ planar \& $G$ pseudoforest}                     & \multicolumn{2}{c}{$O(n)$ }           & {\cite{juenger-sefespqr-09}}                                      \\
    \multicolumn{2}{l}{\row $G$ star}                                                                       & \multicolumn{2}{c}{$O(n)$ }           & {\cite{adfpr-tsegi-12}}                                           \\
    \multicolumn{2}{l}{\row $G$ consists of disjoint cycles}                                                & \multicolumn{2}{c}{$O(n)$ }           & {\cite{br-drpse-13}}                                              \\
    \multicolumn{2}{l}{\row $G$ consists of components with fixed embeddings}                               & \multicolumn{2}{c}{$O(n^2)$ }         & {\cite{br-drpse-13}}                                              \\
    \multicolumn{2}{l}{\row $G$ has maximum degree~3}                                                       & \multicolumn{2}{c}{polynomial }       & {\cite{s-ttp-13}}                                                 \\
    \multicolumn{2}{l}{\row \makebox[5cm][l]{$\1G$ subdiv.\ of triconnected components \& $\2G$ planar}}    & \multicolumn{2}{c}{polynomial }       & {\cite{s-ttp-13}}                                                 \\
    \multicolumn{2}{l}{\row $G$ biconnected}                                                                & \multicolumn{2}{c}{$O(n)$ }           & {\cite{lubiw-testing-10}}                                         \\
    \multicolumn{2}{l}{\row $G$ biconnected}                                                                & \multicolumn{2}{c}{$O(n)$ }           & {\cite{adfpr-tsegi-12}}                                           \\
    \multicolumn{2}{l}{\row $G$ consists of biconnected components}                                         & \multicolumn{2}{c}{polynomial }       & {\cite{s-ttp-13}}                                                 \\
    \multicolumn{2}{l}{\row $\1G$, $\2G$ biconnected \& $G$ connected}                    & \multicolumn{2}{c}{$O(n^2)$ }         & {\cite{br-spoacep-13}}                                                \\
  \bottomrule
  \end{tabular}
  \caption{A list of graph classes that are either known to always have a
    {\sc SEFE} or that contain counterexamples (table at the top).  For the positive examples bounds on the required area and number of bends per edge are given, provided that they are known.  The symbol \yes / \no\ denotes that a complete characterization of positive and negative instances is given.
    The table at the bottom shows results concerning the computational
    complexity of {\sc SEFE}.}
  \label{tab:sefe}
\end{table}

\begin{figure}[p]
  \centering
    \begin{tikzpicture}[overlay]
      \node at (-5.5*16pt, -1pt-4*16pt) [anchor=north]
      {\cite{ek-sepgfb-05}}; 
      \node at (-5.5*16pt, -1pt-0pt) [anchor=north]
      {\cite{beppe-outer-07}}; 
      \node at (6.5*16pt, -1pt-0pt) [anchor=north]
      {\cite{beppe-outer-07}}; 
      \node at (6.5*16pt, -1pt-4*16pt) [anchor=north]
      {\cite{beppe-outer-07}}; 
      \node at (-9.5*16pt, -1pt-8*16pt) [anchor=north]
      {\cite{fab-sefe-06}}; 
      \node at (-3*16pt, -1pt-8*16pt) [anchor=north]
      {\cite{fab-sefe-06}}; 
      \node at (6.75*16pt, -1pt-8*16pt) [anchor=north]
      {\cite{juenger-sefespqr-09}}; 
      \node at (8.25*16pt, -1pt-11.5*16pt) [anchor=north]
      {\cite{juenger-sefeinter-09}}; 
      \node at (-7.5*16pt, -1pt-11.5*16pt) [anchor=north]
      {\cite{sefe-11}}; 
      \node at (0.5*16pt, -1pt-11*16pt) [anchor=north]
      {\cite{sefe-11}}; 
      \node at (9*16pt, -1pt-18*16pt) [anchor=north]
      {\cite{se-original-07}}; 
      \node at (0.5*16pt, -1pt-18*16pt) [anchor=north]
      {\cite{se-original-07}}; 
      \node at (-7.5*16pt, -1pt-15*16pt) [anchor=north]
      {\cite{se-original-07}}; 
      \node at (0.5*16pt, -1pt-15*16pt) [anchor=north]
      {\cite{fab-sefe-06}}; 
      \node at (-3*16pt, -1pt-28*16pt) [anchor=north]
      {\cite{lubiw-testing-10, adfpr-tsegi-12}}; 
      \node at (-8.5*16pt, -1pt-28*16pt) [anchor=north]
      {\cite{adfpr-tsegi-12}}; 
      \node at (7.75*16pt, -1pt-22*16pt) [anchor=north]
      {\cite{juenger-sefe-06}}; 
      \node at (-7*16pt, -1pt-31*16pt) [anchor=north]
      {\cite{juenger-sefespqr-09}}; 
      \node at (4.5*16pt, -1pt-31*16pt) [anchor=north]
      {\cite{juenger-sefespqr-09}}; 
      \node at (5.75*16pt, -1pt-28*16pt) [anchor=north]
      {\cite{br-spoacep-13}}; 
      \node at (-8.5*16pt, -1pt-22*16pt) [anchor=north]
      {\cite{br-drpse-13}}; 
      \node at (-1*16pt, -1pt-22*16pt) [anchor=north]
      {\cite{br-drpse-13}}; 
      \node at (-9*16pt, -1pt-25*16pt) [anchor=north]
      {\cite{s-ttp-13}}; 
      \node at (-2.75*16pt, -1pt-25*16pt) [anchor=north]
      {\cite{s-ttp-13}}; 
      \node at (6.5*16pt, -1pt-25*16pt) [anchor=north]
      {\cite{s-ttp-13}}; 
    \end{tikzpicture}

  \includegraphics{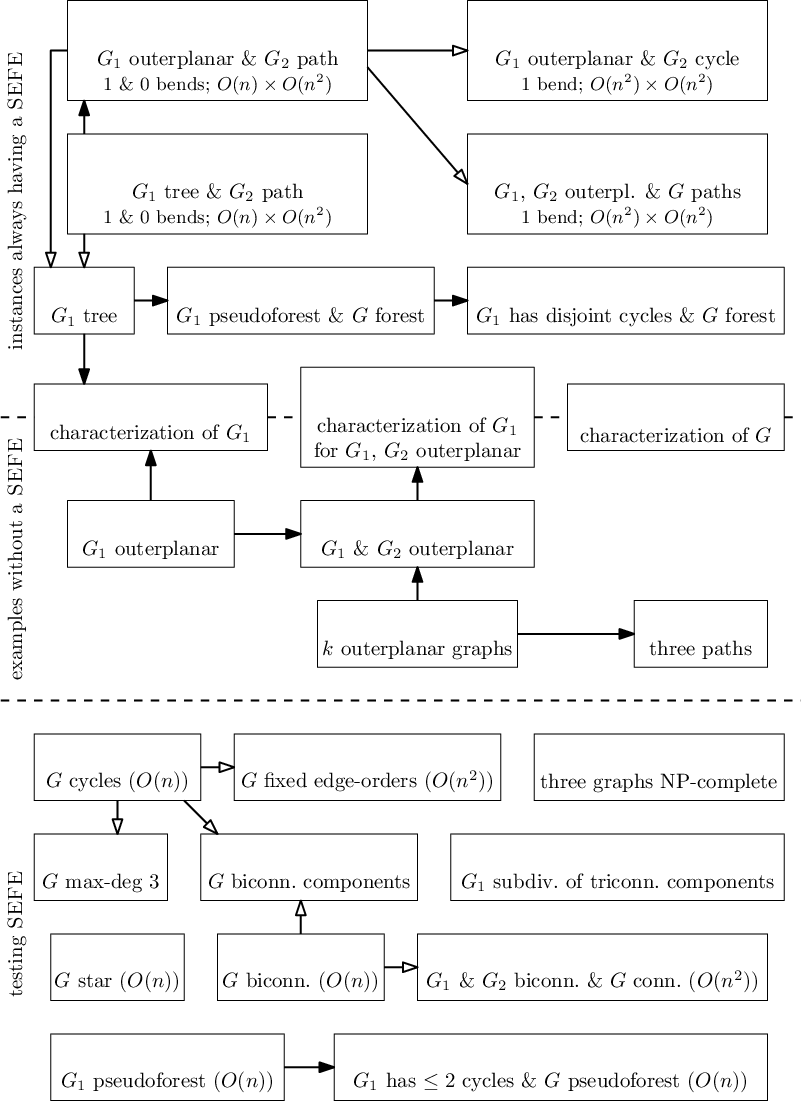}
  \caption{Overview over the so far known results on {\sc SEFE}.  Each box
    represents one result and an arrow highlights that the source-result is
    extended by the target-result.  The arrowheads are empty for the cases in
    which this is only true, if the number of bends per edge, the consumed grid
    size or the necessary running time is neglected.  Note that transitive
    arrows are omitted.
  \label{fig:sefe-overview}
}
\end{figure}
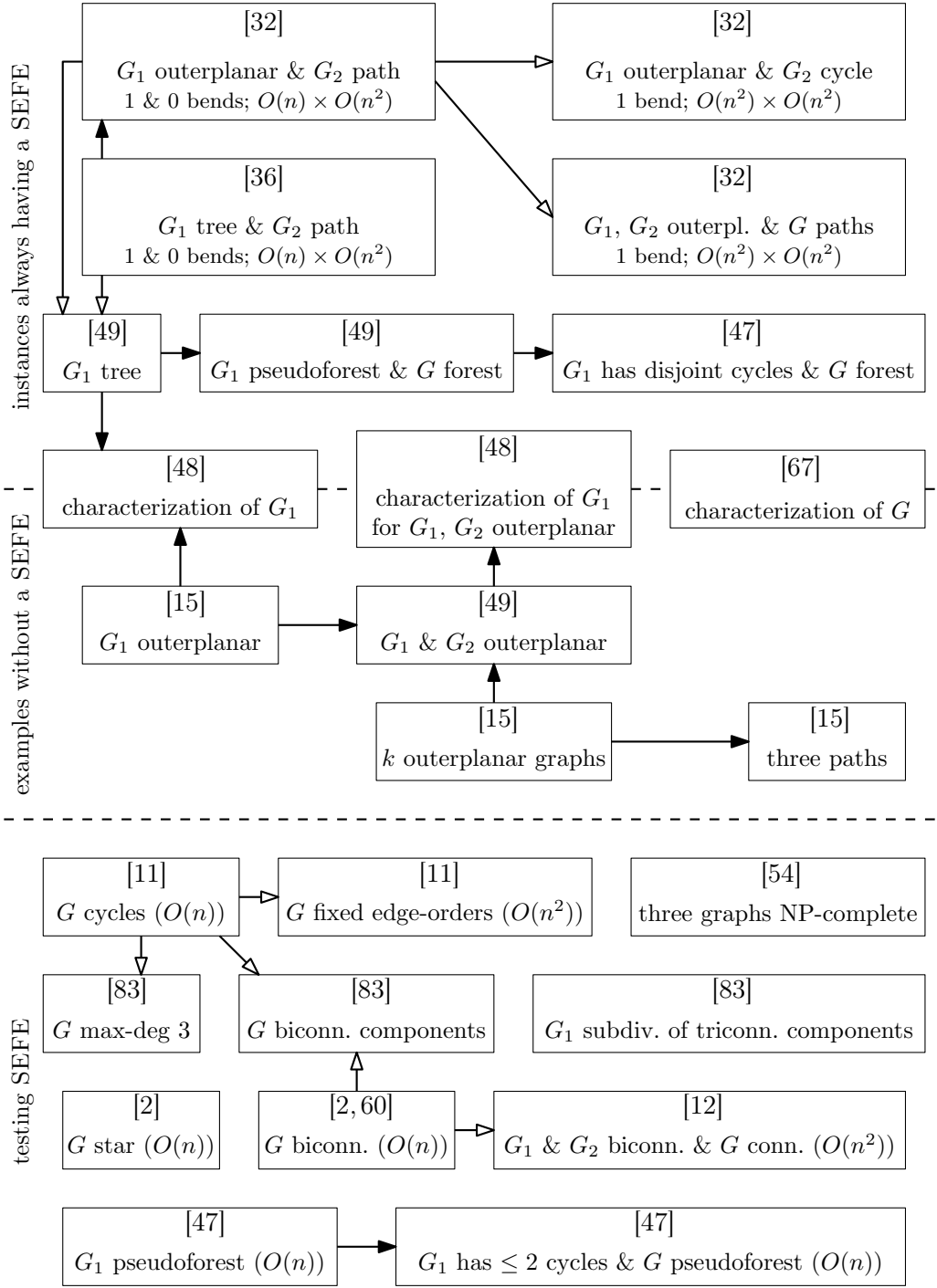

\subsection{Positive and Negative Examples}

We start with instances that always admit a {\sc SEFE}.  Erten and
Kobourov~\cite{ek-sepgfb-05} show that a tree and a path can always be
embedded simultaneously.  They additionally give an algorithm finding
a simultaneous embedding in $O(n)$ time on a grid of size $O(n) \times
O(n^2)$ such that the edges of $\1G$ and $\2G$ have at most one and
zero bends per edge, respectively.  Note that a grid of size $O(n^2)
\times O(n^3)$ is necessary if the bends are required to be drawn on
grid points. Di Giacomo and Liotta~\cite{beppe-outer-07} extend this
result to the case of an outerplanar graph and a path with the same
grid and bend requirements.  They extend it further to the case where
$\1G$ and $\2G$ are outerplanar and the common graph $G$ is a
collection of paths and to the case where $\1G$ is outerplanar and
$\2G$ is a cycle.  However, in both cases a grid of size $O(n^2)
\times O(n^2)$ and up to one bend per edge are required.  If the grid
and bend requirements are completely neglected, the results
considering the pairs tree plus path and outerplanar graph plus path
can be extended to the case where one of the two graphs is a tree.

Frati~\cite{fab-sefe-06} shows how a tree $\1G$ can be simultaneously
embedded with an arbitrary planar graph $\2G$.  This algorithm still
works if $\1G$ contains one additional edge that is not a common edge,
yielding the result that every graph with at most one cycle (a
\emph{pseudoforest}) can be embedded simultaneously with every other
planar graph if the common graph does not contain this cycle.  Fowler
et al.~\cite{juenger-sefespqr-09} extend this result further to the
case where $\1G$ contains only disjoint cycles and the common graph
$G$ does not contain a cycle.

Aside from instances always having a {\sc SEFE}, there are also
examples that cannot be simultaneously embedded.  Brass et
al.~\cite{se-original-07} give examples for $k$ outerplanar graphs,
three paths and an outerplanar graph plus a planar graph not having a
{\sc SEFE}.  The results concerning outerplanar graphs can be extended
to the case where both graphs are outerplanar~\cite{fab-sefe-06}.  

In between the positive and negative results there are some
characterizations stating which instances have a {\sc SEFE} and which
do possibly not.  Fowler et al.~\cite{sefe-11} give a characterization
of the graphs $\1G$ having a {\sc SEFE} with every other planar graph.
This of course extends all results concerning only~$\1G$.  In
particular, the results that a tree can be simultaneously embedded
with every other graph, whereas an outerplanar graph cannot, are
extended.  This characterization essentially requires that $\1G$ must
not contain a subgraph homeomorphic to $K_3$ (a triangle) and an edge
not attached to this $K_3$; see Figure~\ref{fig:sefe-counterexample}
for an example.  The considerations made for this characterization
additionally yield a characterization for the biconnected outerplanar
graphs $\1G$ having a simultaneous embedding with every other
outerplanar graph $\2G$.  This of course extends the result that two
outerplanar graphs possibly do not have a {\sc SEFE}.  

\begin{figure}[tb]
  \centering
  \includegraphics{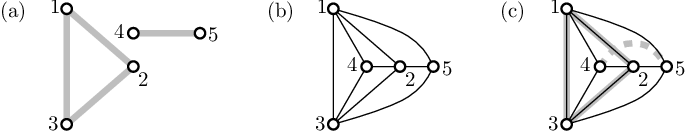}
  \caption{$\1G$ (a) and $\2G$ (b) do not admit a {\sc SEFE} (c) as
    $\2G$ forces the vertices $4$ and $5$ to different sides of the
    triangle $\Delta123$.
  \label{fig:sefe-counterexample}
}
\end{figure}

Another characterization, in terms of the common graph, is given by
J{\"u}nger and Schulz~\cite{juenger-sefeinter-09}.  They show that two
graphs can be simultaneously embedded if the common graph $G$ has only
two embeddings, whereas in all other cases graphs $\1G$ and $\2G$ with
the common graph $G$ not having a {\sc SEFE} can be constructed.  They
additionally show that finding a {\sc SEFE} is equivalent to finding
combinatorial embeddings of $\1G$ and $\2G$ inducing the same
combinatorial embedding, that is the same orders of edges around
vertices and the same relative positions of connected components to
one another, on the common graph~$G$~\cite[Theorem
4]{juenger-sefeinter-09}.  Note that it is not obvious and not even
true for more than two graphs~\cite{fab-embedded-11}.  As this result
is heavily used in most algorithms solving the decision problem {\sc
  SEFE}, we state it as a theorem.

\begin{theorem}
  Two graphs $\1G$ and $\2G$ with common subgraph $G$ admit a {\sc
    SEFE} if and only if they admit combinatorial embeddings inducing
  the same embedding on~$G$.
\end{theorem}

\subsection{Testing SEFE}

Since {\sc SEFE} has positive and negative instances, it would be nice
to have an algorithm deciding for given graphs, whether they can be
embedded simultaneously.  If more than two graphs are allowed, this
problem is known to be NP-complete~\cite{juenger-sefe-06}, whereas the
complexity for two graphs is still open.  However, there are several
results solving {\sc SEFE} for special cases. 

Fowler at al.~\cite{juenger-sefespqr-09} show how to test {\sc SEFE},
if $\1G$ is a pseudoforest, that is a graph with at most one cycle.
Note that, as mentioned above, such an instance always has a {\sc
  SEFE} if this single cycle is not contained in $G$.  This result can
be extended to the case where $\1G$ contains up to two cycles, if $G$
does not contain the second cycle, that is $G$ is a pseudoforest.  To
achieve this result the following auxiliary problem was solved.  Given
a planar graph $G$ with a designated cycle $C$ and a partition
$\mathcal P = \{P_1, \dots, P_k\}$ of the vertices not contained in
$C$, does $G$ admit a planar embedding, such that all vertices in
$P_i$ are on the same side of the cycle for every set $P_i$?  Note
that this again is a constrained embedding problem, showing that
constrained and simultaneous embedding are closely related.  Despite
early effort~\cite{sefe-11}, testing {\sc SEFE} for two outerplanar
graphs remains open.

Haeupler et al.~\cite{lubiw-testing-10} give a linear-time algorithm
to solve {\sc SEFE} for the case that the common graph is biconnected.
Their solution is an extension of the planarity testing algorithm by
Haeupler and Tarjan~\cite{ht-papqt-08}.  This planarity testing
algorithm starts with a completely unembedded graph and adds vertices
iteratively, such that the unembedded part is always connected,
ensuring that it can be assumed to lie in the outer face of all
embedded components.  While inserting vertices, they keep track of the
possible embeddings of the embedded parts by representing the possible
orders of half-embedded edges around every component with a PQ-tree
having these edges as leaves.  In a \emph{PQ-tree} every inner node is
either a Q-node fixing the order of edges incident to it up to a flip
or a P-node allowing arbitrary orders.  In this way a PQ-tree
represents a set of possible orders of its leaves.

A completely different approach is used by Angelini et
al.~\cite{adfpr-tsegi-12} to solve {\sc SEFE} in linear time if the
common graph is biconnected.  They choose an order for the common
graph bottom up in its SPQR-tree such that the private edges can be
added.

Another approach by Bl\"asius and Rutter~\cite{br-spoacep-13}
also uses PQ-trees.  They use that the possible orders of edges around
every vertex of a biconnected planar graph can be represented by a
PQ-tree, yielding a set of PQ-trees, one for each vertex.  To obtain a
planar embedding, the orders for the PQ-trees have to be chosen
consistently.  Bl\"asius and Rutter define the problem {\sc
  Simultaneous PQ-Ordering} asking for orders in PQ-trees that are
chosen consistently, which can, among other applications, be used to
represent all planar embeddings of a biconnected graph.  This extends
to the case of two biconnected planar graphs enforcing shared edges to
be ordered the same and thus yields a quadratic time algorithm for
{\sc SEFE} if $\1G$ and $\2G$ are biconnected and $G$ is connected.
The latter requirement comes from the fact that only orders of edges
around vertices are taken into account, relative positions of
connected components to one another are neglected.  Note that this
result extends the case where $G$ is biconnected for the following
reason.  If $G$ is biconnected, then $G$ is completely contained in a
single block (maximal biconnected component) of $\1G$ and $\2G$.
Thus, even if $\1G$ or $\2G$ are not biconnected, they contain only
one block that is of interest, all other blocks can simply be attached
to this block.  

The result by Bl\"asius and Rutter can be slightly extended to the
case where the graphs $\1G$ and $\2G$ contain cut-vertices incident to
at most two non-trivial blocks (blocks not consisting of a single
edge), including the special case where both graphs have maximum
degree~5.  The {\sc Simultaneous PQ-Ordering} approach again shows the
strong relation between simultaneous and constrained embedding as in
an instance of {\sc SEFE} the two input graphs constrain the possible
orders of some of the edges around vertices of one another with
PQ-trees.

Angelini et al.~\cite{adfpr-tsegi-12} show the equivalence between
{\sc SEFE} and a constrained version of the {\sc Partitioned 2-Page
  Book Embedding} problem. An instance of {\sc Partitioned 2-Page Book
  Embedding} is a graph and a partition of its edges into two subsets.
It asks whether all vertices can be arranged on a straight line (the
\emph{spine}) such that each of the edge partitions can be embedded
without crossings in one of the two incident half-planes (\emph{pages}
of the book).  {\sc Partitioned $T$-Coherent 2-Page Book Embedding}
additionally has a tree as input with the vertices of the graph as
leaves.  It is then required that the tree admits an embedding such
that the order of its leaves is equal to the order of vertices on the
spine.  In other words, the allowed orders of vertices on the spine is
constrained by a PQ-tree containing no Q-nodes.  Angelini et
al.~\cite{adfpr-tsegi-12} prove the following theorem and we sketch
their proof here.

\begin{theorem}
  The problems {\sc SEFE} for two graphs with connected intersection
  and {\sc Partitioned $T$-Coherent 2-Page Book Embedding} have the
  same time complexity.
\end{theorem}
\begin{sketch}
  Angelini et al.~\cite{adfpr-tsegi-12} first show that an instance of
  {\sc SEFE} where the common graph is connected can be modified
  (yielding an equivalent instance) such that the common graph is a
  tree.  Moreover, each private edge is incident to leaves of this
  tree.  They then show the equivalence to an instance of {\sc
    Partitioned $T$-Coherent 2-Page Book Embedding} where the common
  graph is the constraining tree, the leaves of this tree are the
  vertices that need to be placed on the spine and the private edges
  of each of the graphs is one of the partitions.  

  In the following we sketch this construction using the example in
  Figure~\ref{fig:sefe-book-emb}.  The instance in~(a) having a tree
  $T$ as common graph such that each private edge is incident to a
  leaf admits a {\sc SEFE}.  All private edges are embedded outside
  the dashed cycle around $T$ in~(b) containing all its leaves.
  Choosing another face as outer face and cutting the cycle at an
  arbitrary position yields a {\sc SEFE} where all leaves of $T$ are
  embedded on a straight line~(c) with all private edges on the same
  side.  This directly yields the {\sc Partitioned $T$-Coherent 2-Page
    Book Embedding} in~(d) of the private edges respecting the tree
  $T$.  This shows the equivalence of {\sc SEFE} and {\sc Partitioned
    $T$-Coherent 2-Page Book Embedding} as the constructions works the
  same in the opposite direction.  
\end{sketch}

  \begin{figure}[tb]
    \centering
    \includegraphics{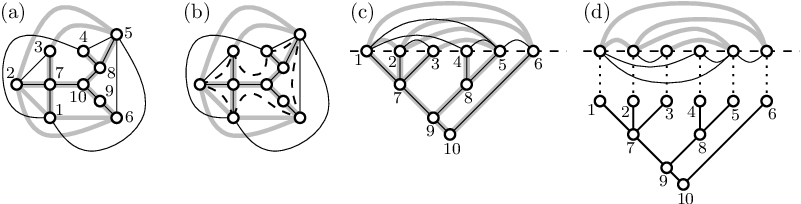}
    \caption{Equivalence of an instance of {\sc SEFE} and the
      corresponding instance of {\sc Partitioned $T$-Coherent 2-Page
        Book Embedding}. 
    \label{fig:sefe-book-emb}
}
  \end{figure}

For the restricted case that $T$ is a star, {\sc Partitioned
  $T$-Coherent 2-Page Book Embedding} reduces to the problem {\sc
  Partitioned 2-Page Book Embedding} that can be solved in linear
time~\cite{hn-tpbecgp-09}.  Thus the above result directly implies
that {\sc SEFE} can be solved in linear time if the common graph is a
star.

All results mentioned thus far require $G$ to be connected and most
results also require $\1G$ and $\2G$ to be connected.  Bl{\"a}sius and
Rutter~\cite{br-drpse-13} consider the case where this does not hold.
They show that it can be assumed without loss of generality that both
graphs $\1G$ and $\2G$ are connected.

In the case that $G$ is connected, one only has to deal with orders of
edges around vertices and can neglect relative positions of connected
components to one another.  Bl{\"a}sius and Rutter approach {\sc SEFE}
from the opposite direction, caring only about the relative positions,
neglecting the orders of edges around vertices.  More precisely, they
give a linear-time algorithm solving {\sc SEFE} if the common graph is
a set of disjoint cycles.  They can extend this result to a
quadratic-time algorithm for the case where $G$ consists of arbitrary
connected components, each with a fixed planar embedding.  Both
results extend to an arbitrary number of graphs with sunflower
intersection.  Recall that sunflower intersection means that all
graphs intersect in the same common subgraph.  Moreover, they give a
succinct representation of all simultaneous embeddings.

\todo{new result by Schaefer added} A completely different, algebraic
approach is presented by Schaefer~\cite{s-ttp-13}.  It is based on the
Hanani-Tutte theorem~\cite{h-uwukr-34,t-ttcn-70} stating that a graph
is planar if and only if its \emph{independent odd crossing number}
is~0.  The independent odd crossing number of a drawing is the number
of pairs of non-adjacent edges whose number of crossings is odd.  The
independent odd crossing number of a graph is its minimum over all
drawings.  Thus, by the Hanani-Tutte theorem, testing planarity is
equivalent to testing whether this crossing number is~0.  The latter
condition can be formulated as a system of linear equations over the
field of two elements, leading to a simple polynomial-time planarity
algorithm.  Schaefer extends this result to other notions of
planarity.  In particular, it is shown that {\sc SEFE} can be solved
in polynomial time for three interesting cases, namely (1) if the
common graph $G$ consists of disjoint biconnected components and
isolated vertices, (2) if the common graph has maximum degree~3, and
(3) if $\1G$ is the disjoint union of subdivisions of triconnected
graphs.  When neglecting the slower running time, this extends several
of the results known before; see Figure~\ref{fig:sefe-overview}.

\subsection{Related Work}

A result not really fitting in one of the three above classes by
Duncan et al.~\cite{dek-gtldg-04} considers the restricted case of
{\sc SEFE} where each edge has to be a sequence of horizontal and
vertical segments with at most one bend per edge.  They show that two
graphs with maximum degree~2 always admit such a {\sc SEFE} on a grid
of size $O(n) \times O(n)$ by adapting their linear-time algorithm
computing an {\sc SGE} for these types of graphs (on a larger grid).

Angelini et al.~\cite{fab-embedded-11} consider the case where the
embedding of each of the input graphs is already fixed.  With this
restriction {\sc SEFE} becomes trivial for two graphs since it remains
to test whether the two graphs induce the same embedding on the common
graph.  They show that it can also be decided efficiently for three graphs.
However, it becomes NP-hard for at least fourteen graphs.  They also
consider the problem {\sc SGE} for the case that the embedding of each
graph is fixed and show that it is NP-hard for at least thirteen
graphs.

Schaefer~\cite{s-ttp-13} shows that several other notions of planarity
are related to {\sc SEFE}.  In particular, the well-studied cluster
planarity problem reduces to {\sc SEFE}, providing further incentive
to study its complexity.

\section{Simultaneous Embedding}
\label{sec:se}

In the most restricted version of the problem, {\sc SGE}, we insist
that vertices are placed in the same position, and edges must be
straight-line segments. The {\sc SEFE} setting relaxes the
straight-line condition but maintains that edges common to multiple
graphs are realized the same way in each. In the least restrictive
setting, {\sc SE}, we allow the same edge to be realized differently
in different graphs.

It has already been mentioned that simultaneous embedding of multiple
graphs can be thought of as a generalization of the notion of
planarity. A classical result about planar graphs connects the notion
of a planar graph with that of a straight-line, crossing-free drawing
thereof. Specifically, Wagner in 1936~\cite{w-bzv-36}, F{\'a}ry in
1948~\cite{f-slrpg-48}, and Stein in 1951~\cite{s-cm-51} independently
show that if a graph has a drawing without crossings, using arbitrary
curves as edges, then there exists a drawing of the graph also without
crossings, but with edges drawn as straight-line segments. For
multiple graphs, however, this result does not hold. That is, given
several graphs on the same $n$ vertices, we can surely realize each
graph without crossings, using arbitrary curves as edges and the same
vertex positions for each graph. But (except in very special
circumstances such as the positive examples in the
Section~\ref{sec:sge}) we cannot guarantee that there exist vertex
positions that allow the realization of each graph with straight-line
segments and without crossings. If this were true, then the vertex
positions would be a {\em universal pointset} for graphs on $n$
vertices, and it is known that universal pointsets of linear size do
not exist~\cite{fpp-hdpgg-90}.

Pach and Wenger~\cite{pw-epgfvl-98} show that every planar graph can
be drawn without crossings with a prespecified position for every
vertex.  Thus, for every pair of planar graphs an {\sc SE} can be
created by drawing the first graph arbitrarily and the second graph to
the vertex positions specified by the first drawing.  Thus, there are
neither negative examples nor is it necessary to have testing
algorithms.  However, the drawing of the second graph may have
linearly many bends per edge, thus it is of interest to find an {\sc
  SE} with fewer bends. 

Erten and Kobourov~\cite{ek-sepgfb-05} show that every two graphs can
be drawn simultaneously in $O(n)$ time with at most three bends per
edge on an $O(n^2) \times O(n^2)$ grid ($O(n^3) \times O(n^3)$ if
bends need to be placed on grid points), where $n$ is the number of
vertices.  \todo{theorem + proof + figure added} To achieve this
result, they combine the construction of Brass et
al.~\cite{se-original-07} to create an {\sc SGE} of two paths (see
Theorem~\ref{the-two-path-theorem} in Section~\ref{sec:sge}) with a
technique by Kaufmann and Wiese~\cite{kw-evpfb-02}, who show that
every planar graph can be drawn with at most two bends per edge if the
allowed vertex positions are restricted to a set of points.  We
include the main result from this paper along with a proof sketch.

\begin{figure}[tb]
  \centering
  \includegraphics{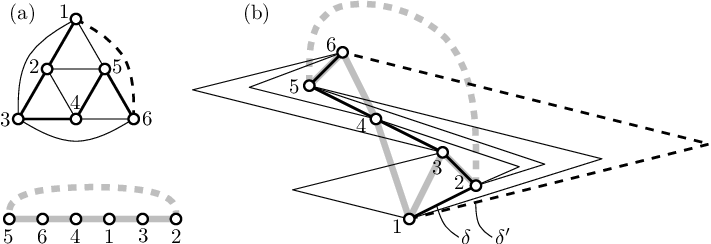}
  \caption{(a)~The cycle $\2H$ (gray) with the path $\2P$ (not dashed)
    and the graph $\1G$ containing the Hamiltonian cycle $\1H$ (bold)
    and the Hamiltonian path $\1P$ (bold, not dashed).  (b)~The
    drawing of $\1G$ and $\2P$ according to the construction of
    Theorem~\ref{thm:se-few-bends}.
    \label{fig:se-few-bends}
  }
\end{figure}

\begin{theorem}
  \label{thm:se-few-bends}
  For two planar graphs $\1G$ and $\2G$ an {\sc SE} with at most three
  bends per edge on an $O(n^2) \times O(n^2)$ grid can be found in
  linear time.
\end{theorem}
\begin{sketch}
  Initially, assume that $\1G$ and $\2G$ are 4-connected.  This
  assumption is removed later using the technique of Kaufmann and
  Wiese.

  We can compute Hamiltonian cycles $\1H$ and $\2H$ of $\1G$ and
  $\2G$, respectively, using the algorithm of Chiba and
  Nishizeki~\cite{cn-hcp-89}.  Let $\1P$ and $\2P$ be Hamiltonian
  paths contained in $\1H$ and $\2H$, respectively; see
  Figure~\ref{fig:se-few-bends}(a) for an example.  As in the proof of
  Theorem~\ref{the-two-path-theorem}, we can construct an {\sc SGE} of
  $\1P$ and $\2P$ such that $\1P$ is $y$-monotone, while $\2P$ is
  $x$-monotone.  We show how to add the remaining edges of $\1G$ and
  the construction is similar for $\2G$.

  We consider the absolute values of the slopes the edges in $\1P$
  have and define $\delta$ to be their minimum.  Let further $\delta'$
  be slightly smaller.  We first close the cycle $\1H$ by adding the
  missing edge using two straight-line segments with slopes $\delta'$
  and $-\delta'$; see Figure~\ref{fig:se-few-bends}(b).  Similarly,
  all remaining edges of $\1G$ are drawn with two straight-line
  segments with slopes appropriately chosen between $\delta'$ and
  $\delta$ and between $-\delta$ and $-\delta'$.  Dealing similarly
  with the remaining edges of $\2G$ yields an {\sc SE} with at most
  one bend per edge on a grid of size $O(n^2) \times O(n^2)$.

  For the case that $\1G$ and $\2G$ are not 4-connected, Kaufmann and
  Wiese~\cite{kw-evpfb-02} showed how they can be augmented to
  4-connected planar graphs by adding new edges and subdividing every
  edge at most once.  Drawing these augmented graphs as described
  above, removing the additional edges and replacing each subdivision
  vertex with a bend yields an {\sc SE} of $\1G$ and $\2G$ with at
  most three bends per edge on an $O(n^2) \times O(n^2)$ grid.
\end{sketch}

The result of Erten and Kobourov was improved by Di Giacomo and
Liotta~\cite{beppe-note-05,beppe-outer-07} to at most two bends per
edge in general and one bend per edge, if $\1G$ and $\2G$ are both
sub-Hamiltonian.  That is, they can be augmented to become Hamiltonian
maintaining planarity, and an augmentation together with a Hamiltonian
cycle is given with the input.  Similar results were obtained by
Kammer~\cite{kammer-06}.  As series-parallel
graphs~\cite{ddlw-bespd-06}, trees and outerplanar
graphs~\cite{clr-egb-87, bk-btg-79} are always sub-Hamiltonian and an
augmentation together with a Hamiltonian cycle can be computed in
linear-time this result yields a linear time algorithm to compute an
{\sc SE} of $\1G$ and $\2G$ with one bend per edge on a grid of size
$O(n^2) \times O(n^2)$ if each of the graphs $\1G$ and $\2G$ is
series-parallel, a tree or outerplanar.

Cappos et al.~\cite{se-arcs-09} show that a path and an outerplanar
graph can be simultaneously embedded in linear time such that edges in
the outerplanar graph are straight-line segments and each edge in the
path consists of a single circular arc.  Alternatively, the path edges
may be piecewise linear with at most two bends per edge.


\section{Colored Simultaneous Embedding}
\label{sec:color-se}

Since {\sc SGE} can be too restrictive, various relaxations have been
considered. The two relaxed versions already mentioned, {\sc SEFE} and
{\sc SE} relax the requirement of straight-line edges, and even the
requirement that common edges are drawn the same way in both
drawings. Another way to relax the constraints of the original {\sc
  SGE} problem is to allow changes in vertex positions in different
graphs.

Until this point we had assumed that multiple input graphs have
labeled vertices and thus the mapping between the vertices of the
graphs is part of the input.  In {\em simultaneous embedding without
  mapping} we are interested in computing plane drawings for each of
the given graphs on the same set of points, where any vertex can be
placed at any of the points in the point set. This setting of the
problem was investigated in the very first paper on {\sc
  SGE}~\cite{se-original-07} and is the source of one of the longest
standing open problems in the area.

\begin{figure}[tb]
  \centering
  \includegraphics{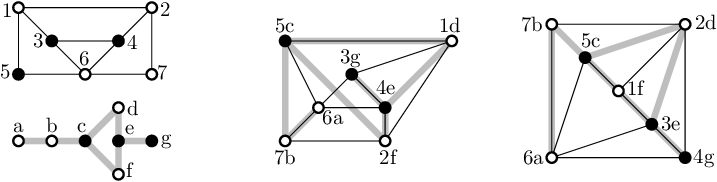}
  \caption{Two 2-colored graphs with two {\sc CSE}s corresponding to
    different mappings. \label{fig:cse}}
\end{figure}

A common generalization of the problems above is {\sc Colored
  Simultaneous Embeddings} ({\sc CSE}), which was introduced by
Brandes et al.~\cite{colored-se-11}, and contains both, the version
with and without mapping. Formally, the problem of {\sc CSE} is
defined as follows. The input is a set of planar graphs $\1G =
(V,\1E), \2G = (V,\2E), . . . , \k G = (V,\k E)$ on the same vertex
set $V$ and a partition of $V$ into $c$ classes, which we refer to as
colors. The goal is to find plane straight-line drawings $\ii D$ of
$\ii G$ using the same $|V|$ points in the plane for all $i = 1,
\ldots k$, where vertices mapped to the same point are required to be
of the same color.  We call such graphs $c$-colored graphs;
\todo{figure added}see Figure~\ref{fig:cse} for an example. Given the
above definition, simultaneous embeddings with and without mapping
correspond to colored simultaneous embeddings with $c = |V|$ and $c =
1$, respectively. Thus, when a set of input graphs allows for a
simultaneous embedding without mapping but does not allow for a
simultaneous embedding with mapping, there must be a threshold for the
number of colors beyond which the graphs can no longer be embedded
simultaneously.

Colored simultaneous embeddings provide a way to obtain
near-simultaneous embeddings, where we place corresponding vertices
nearly, but not necessarily exactly, at the same locations. Relaxing
the constraint on the size of the pointset allows for a way to more
easily obtain near-simultaneous embeddings, where we attempt to place
corresponding vertices relatively close to one another in each
drawing. For example, if each cluster of points in the plane has a
distinct color, then even if a red vertex $v$ placed at a red point
$p\in \1G$ has moved to another red point $q \in \2G$, the movement is
limited to the area covered by the red points.  

Brandes et al.~\cite{colored-se-11} show several positive and negative
results about {\sc CSE}. In particular they show that there exist
universal pointsets of size $n$ for 2-colored paths and spiders as
well as 3-colored paths and caterpillars.  It is also shown that a
2-colored tree (or even a 2-colored outerplanar graph) and any number
of 2-colored paths can be simultaneously embedded.  In the negative
direction, there exist a 2-colored planar graph and pseudo-forest,
three 3-colored outerplanar graphs, four 4-colored pseudo-forests,
three 5-colored pseudo-forests, five 5-colored paths, two 6-colored
biconnected outerplanar graphs, three 6-colored cycles, four 6-colored
paths, and three 9-colored paths that cannot be simultaneously
embedded.

Frati et al.~\cite{se-constrained-09} continue the investigation of
near-{\sc SGE}'s, that is they try to find straight-line drawings of
the input graphs with a small distance between every pair of common
vertices in different drawings.  As a negative result, they present a
pair of graphs such that in every pair of drawings there exists a
common vertex with distance linear in the size of the input.  On the
other hand, they present positive results for a sequence of paths and
a sequence of trees for the case that every two consecutive graphs in
the sequence are similar with respect to a parameter measuring their
similarity.  It can then be shown that the distance of a common vertex
in two consecutive drawings depends linearly on this parameter.

\section{Matched Drawings}
\label{sec:matched-drawings}

Another approach to relax requirements of {\sc SGE} are the so-called
\emph{matched drawings} introduced by Di Giacomo et
al.~\cite{ddkls-mdpg-09}.  A matched drawing of a pair of graphs is a
planar straight-line drawing of each of the graphs such that each
common vertex has the same $y$-coordinate in both drawings (instead of
the same $y$- and $x$-coordinate as required for {\sc SGE});
\todo{figure added}see Figure~\ref{fig:matched-drawings} for an
example.  

\begin{figure}[tb]
  \centering
  \includegraphics{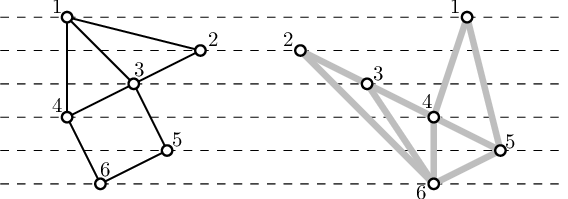}
  \caption{A matched drawing: corresponding vertices have the same
    $y$-coordinate. \label{fig:matched-drawings}}
\end{figure}

Di Giacomo et al.~\cite{ddkls-mdpg-09} give a small counterexample
consisting of two small triconnected planar graphs not admitting a
matched drawing.  Moreover, they give a larger example (620 vertices)
of a biconnected graph and a tree not having a matched drawing.  

Apart from that they also have some results on the positive side.
They show that two trees are always matched drawable.  Moreover, they
observe that any planar graph has a matched drawing with a so-called
\emph{unlabeled level planar (ULP)} graph, that is a graph that admits
a planar straight-line drawing even if the y-coordinate of each vertex
is prespecified such that no two vertices have the same y-coordinate.
A characterization of ULP graphs is given by Fowler and
Kobourov~\cite{fk-culpg-08}.  Di Giacomo et al.~\cite{ddkls-mdpg-09}
moreover show for a graph class containing non-ULP graphs (the
\emph{carousel graphs}) that they admit matched drawings with
arbitrary planar graphs.  A special cases of a carousel graphs is a
graph consisting of a single vertex $v_0$ and a set of disjoint
subgraphs $S_1, \dots, S_k$, each $S_i$ connected to $v_0$ over a
single edge $\{v_0, v_i\}$ such that $S_i$ is either a caterpillar
with $v_i$ on its spine, a radius-2 star with $v_i$ as center or a
cycle.

Grilli et al.~\cite{ghlmw-mdgpgt-09} present further positive results
on matched drawings.  They show how to draw the pairs outerplane plus
wheel, wheel plus wheel, outerplane plus \emph{maximal outerpillar}
(outerplane graph with triangulated inner faces and caterpillar as
weak dual), and outerplane plus \emph{generalized outerpath}
(outerpath where some edges on the outer face may be replaced by some
small subgraphs).  Moreover, they consider matched drawings for graph
triples and give algorithms creating matched drawings of three cycles,
and a caterpillar and two ULP graphs.

\section{Other Simultaneous Representations}
\label{sec:other-simult-repr}

\todo{introduction added} Apart from simultaneously drawing two graphs
sharing some common parts there are other ways to represent graphs
simultaneously.  In this section we describe how a plane graph and its
dual can be represented simultaneously, and what is known about
simultaneous intersection representations of (not necessarily planar)
graphs.

\subsection{A Plane Graph and its Dual}

In a simultaneous drawing of a planar graph and its dual each vertex
in the dual graph is required to be placed inside the corresponding
face of the primal graph.  Moreover, no crossings are allowed except
for crossings between a dual and its corresponding primal edge.
Tutte~\cite{t-hdg-63} first considered this problem and showed that
every triconnected planar graph admits a simultaneous straight-line
drawing with its dual.  However, the resulting drawings may have
exponentially large area.  Erten and Kobourov~\cite{se-dual-05}
provide a linear-time algorithm simultaneously embedding a
triconnected planar graph and its dual on a grid of size $(2n-2)
\times (2n-2)$ such that all edges are drawn as straight-line
segments.  \todo{result by Zhang and He added}Zhang and
He~\cite{zh-sslge-06} improved this result to a grid of size $(n-1)
\times n$.

Brightwell and Scheinerman~\cite{bs-rpg-93} show the existence of a
simultaneous straight-line drawing of a triconnected planar graph and
its weak dual such that the crossings between dual and the
corresponding primal edges are right-angular crossings.  \todo{circle
  packings added} A \emph{circle packing} of a planar graph represents
the vertices as non-crossing circles such that two vertices are
adjacent if and only if their corresponding circles touch.  Given a
circle packing of a planar graph, one obtains a planar straight-line
drawing by placing each vertex at the center of its corresponding
circle.  Mohar~\cite{m-cpmpt-97} shows that every triconnected planar
graph has a simultaneous circle packing with its dual such that in the
corresponding straight-line drawings primal and dual edges have
right-angular crossings.  Argyriou et al.~\cite{abks-grsdg-12} give a
simple example of a graph that is not triconnected not admitting such
a drawing.  On the positive side they give an algorithm that creates
such drawings for the case that the primal graph is outerplanar.

\todo{new part about tessellation representation added}Another way of
simultaneously representing a planar graph and its dual is the
\emph{tessellation representation} introduced by Tamassia and
Tollis~\cite{tt-trpg-89}.  In a tessellation representation, every
edge, every vertex and every face is represented by a (maybe
degenerated) rectangle, a so called \emph{tile}, such that the
interiors of these tiles are pairwise disjoint, that their union forms
a rectangle, and that the incidences in the graph are represented by
side contacts of the tiles in the following way.  Two tiles share a
horizontal line segment if and only if they represent an edge and an
incident face, and two tiles share a vertical line segment if and only
if they represent an edge and an incident vertex.  Tamassia and
Tollis~\cite{tt-trpg-89} in particular showed that every biconnected
planar graph admits a tessellation representation where the tiles
representing vertices and faces are vertical and horizontal line
segments, respectively.  The textbook by Di Battista et
al.~\cite[Sections~4.3 \& 4.4]{dett-gdavg-99} contains a short
description of the algorithm computing tessellation representations
and of the relation to visibility representations.  Moreover,
tessellation representations were also considered on other surfaces
such as the torus~\cite{mr-tvrmt-98}.

\subsection{Intersection Representations}

Jampani and Lubiw~\cite{lubiw-chordal-09} introduce the concept of
simultaneous graph representations for other representations than
drawings.  An intersection representation of a graph assigns a
geometric object to each vertex such that two vertices are adjacent if
and only if their corresponding geometric objects intersect.  Two
graphs sharing a common subgraph are simultaneous intersection graphs
if each of them has an intersection representation such that the
common vertices are represented by the same objects.  Note that every
planar drawing of a graph can be interpreted as intersection
representation, each vertex is represented by the union of its edges.
This shows that deciding {\sc SEFE} as a special case of recognizing
simultaneous intersection graphs.

Other popular intersection representations are the following.  In an
\emph{interval representation} of a graph each vertex is represented
by an interval on the real line.  A graph is \emph{chordal} if each
induced cycle has length three.  Gavril~\cite{g-igstcg-74} shows that
chordal graphs are exactly the intersection graphs of subtrees in a
tree.  This shows that the class of interval graphs is contained in
the class of chordal graphs.  \emph{Permutation graphs} are the
intersection graphs that can be represented by a set of line segments
connecting two parallel lines.  Jampani and
Lubiw~\cite{lubiw-chordal-09} give $O(n^3)$-time algorithms
recognizing simultaneous permutation graphs and simultaneous chordal
graphs.  The algorithm for simultaneous permutation graphs can be
extended to more than two graphs with sunflower intersection.  On the
other hand, it is NP-hard to recognize simultaneous chordal graphs of
this kind (for a constant number $k$ of graphs, the complexity is
still open).

In a follow-up paper Jampani and Lubiw~\cite{lubiw-interval-10} give
an algorithm recognizing simultaneous interval graphs in $O(n^2\log
n)$ time.  As interval graphs can be characterized in terms of
PQ-trees, recognizing simultaneous interval graphs leads to a problem
of finding orders in several PQ-trees simultaneously.  Bl\"asius and
Rutter~\cite{br-spoacep-13} consider this kind of problem in a more
general leading to a $O(n)$-time algorithm recognizing simultaneous
interval graphs.

Related to simultaneous intersection graphs are simultaneous
comparability graphs also introduced by Jampani and
Lubiw~\cite{lubiw-chordal-09}.  A \emph{comparability graph} is a
graph that can be oriented transitively where transitively means that
a directed path implies the existence of a directed edge.  Two graphs
are \emph{simultaneous comparability graphs} if each of them can be
oriented transitively such that common edges are oriented the same in
both.  Jampani and Lubiw give an $O(nm)$-time algorithm recognizing
simultaneous comparability graphs.  It can also be used to recognize
an arbitrary number of comparability graphs with sunflower
intersection.  Comparability graphs are related to intersection graphs
as comparability graphs are exactly the graphs whose complement is a
\emph{function graph}, that is the intersection graph with respect to
continuous functions on an interval~\cite{gru-cgig-83}.

As for the problem {\sc SEFE}, finding simultaneous representations is
related to extending a representation of a subgraph to one of the
whole graph.  For interval graphs Klav\'ik et al.~\cite{kkv-eprig-11}
give a $O(nm)$-time algorithm testing whether a partial interval
representation can be extended.  Bl\"asius and
Rutter~\cite{br-spoacep-13} were able to improve the running time to
$O(m)$ by constructing a second graph such that both graphs are
simultaneous interval graphs if and only if the partial interval
representation can be extended.

\section{Practical Approaches to Dynamic Graph Drawing}
\label{sec:practical-approaches}

The majority of the results reviewed above focused on the
theoretical \todo{section extended, added results on dynamic graph
  drawing} aspects of dynamic graph drawing.  In this section we
review practical approaches to this problem.  As we have seen in the
previous sections, numerous negative results show that in many of the
interesting settings we cannot guarantee simultaneous embeddings.  On
the other hand, several efficient algorithms for different variants of
the problem do exists, but they usually place additional restrictions
on the number of input graphs, or limit the graphs to special
sub-classes of planar graphs.

As discussed in the introduction, the problem is well-motivated in
practice.  Of particular interest are applications to visualization of
dynamic graphs and the related issues of mental map preservation and
good graph readability.  With this in mind we mention several more
practical results here.  First we focus on drawing algorithms that aim
to produce simultaneous embeddings or layouts that are in some sense
close to being a simultaneous embedding.  Afterwards, we briefly
discuss other approaches to dynamic graph drawing.

Erten et al.~\cite{se-schemes-05} adapt force-directed algorithms to
create drawings of a series of graphs sharing subgraphs finding a
tradeoff between nice drawings and similarities of common parts.
Kobourov and Pitta~\cite{se-interactive-04} describe an interactive
system which allows multiple users to interactively modify a pair of
graphs simultaneously using a multi-user, touch-sensitive input
device.  While those two approaches focus on straight-line drawings
(corresponding to {\sc SGE}), the GraphSET system by
Estrella-Balderrama et al.~\cite{graphset-10} also allows edges to
have bends.  GraphSET is a tool helping the user to investigate the
theoretical problems {\sc SGE} and {\sc SEFE} and it contains
implementations of several testing and drawing algorithms.  Chimani et
al.~\cite{juenger-crossings-08} create simultaneous drawings of graphs
by drawing the union of the graphs.  Their objective is to minimize
the number of crossings in the drawing, where crossings between edges
of different graphs do not count, yielding a simultaneous embedding if
and only if the number of crossings is zero.

Misue et al.~\cite{mels-lamm-95} initiated the study of drawing
dynamically changing graphs and first proposed several models to
capture the notion of preserving the user's mental map.  In particular
they suggested preservation of orthogonal orderings, proximity
relations, or the topology as a formalization.  Bridgeman and
Tamassia~\cite{bt-dmiog-98} describe and evaluate difference metrics
that are specialized to orthogonal graph drawings.  Purchase et
al.~\cite{phg-himm-07} provide empirical evidence that preserving the
user's mental map indeed assists in comprehending the evolving graph.
Purchase and Samra~\cite{ps-ear-08} argue that for minimizing the node
movement, finding a trade-off is worse than either keeping the exact
node positions or just layouting the next graph from scratch for
memorizing tasks.  In a recent study, Archambault and
Purchase~\cite{ap-mmphuodg-13} observed positive effects of mental map
preservation for localization tasks, both in terms of speed and
accuracy.  Sallaberry et al.~\cite{smm-cvnldg-13} consider mental map
preservation for large graphs and argue that restricting node
movements to small distances is not sufficient for this case.  They
propose to cluster nodes into groups that perform the same movement in
order to increase the stability of the drawing.

Bridgeman et al.~\cite{bfgtv-igaio-97} present InteractiveGiotto, a
bend-minimization algorithm for orthogonal drawings that is designed
for dynamic and interactive scenarios.  Their algorithm supports
arbitrary graph changes and preserves the embedding, all edge
crossings, and the bends of edges.  

Brandes and Wagner~\cite{bw-bpdgl-97} suggest a Bayesian framework for
dynamic graph drawing that can in principal we applied to all layout
styles and allows to choose a trade-off between quality and stability.
Diel and G\"org~\cite{dg-gtac-02} introduce foresighted layouts, where
the basic idea is to layout the union of the graph over all time steps
and to combine vertices and edges whose life times are disjoint, in
order to reduce the size of the drawing.  This automatically
guarantees a high stability of the layout, but possibly incurs a
negative impact on the quality of individual drawings.  G\"org et
al.~\cite{gbpd-dgdso-04} enhance this method by an additional step
that improves the quality of the individual layouts while keeping them
close to the foresighted layout.

North and Woodhull~\cite{nw-ohgd-02} propose a heuristic for online
hierarchical graph drawing by dynamizing the classical Sugiyama
algorithm~\cite{stt-mvuhss-81}.  Collberg et al.~\cite{cknpw-sgbve-03}
describe a system for visualizing the evolution of software
based on force-directed methods applied to so-called time-sliced
graphs.  A time-sliced graph consists of disjoint copies of the graph
at each point in time together with time-slice edges, which connect
corresponding vertices from different points in time.  The algorithm
attempts to place vertices that are connected by a time-slice edge in
roughly the same position.  Frishman and Tal~\cite{ft-odgd-08}
describe an algorithm for online dynamic graph drawing that can be
implemented to run on a GPU.

\section{Morphing Planar Drawings}
\label{sec:morphing}

\todo{new section on morphing}The main motivation for simultaneously
embedding different (but related) graphs is to preserve the mental map
between the unchanged parts by drawing them the same.  As opposed to
this, \emph{morphing} tries to match different drawings of the same
graph.  More precisely, let $\Gamma_1$ and $\Gamma_2$ be two drawings
of the same graph $G$, a morph between them is a motion of the
vertices along trajectories starting at the vertex positions in
$\Gamma_1$ and ending at their positions in $\Gamma_2$.

The simplest possible morph between two drawings $\Gamma_1$ and
$\Gamma_2$ is the \emph{linear morph} where each vertex moves at
constant speed along a line segment from its origin in $\Gamma_1$ to
its destination in $\Gamma_2$.  However, the intermediate drawings of
linear morphs may be pretty bad, in fact, it may even happen that the
whole graph collapses to a single point.  To resolve this problem
Cairns~\cite{c-dprc-44} introduced the notion of morphing planar
graphs, requiring that every intermediate drawing is also planar.  He
showed that two planar drawings of a triangulated plane graph with an
equally drawn outer face can be morphed into each other in a planar
way using a sequence of linear morphs.  However, this sequence of
linear morphs has exponential size.

Thomassen~\cite{t-dpg-83} extends this to drawings of general (not
necessarily triangulated) planar graphs with an equally drawn outer
face and convex faces by augmenting the drawings to \emph{compatible
  triangulations}, that is one must be able to add all new vertices
and edges to both given drawings without violating the planarity or
straight-line requirement.  Compatible triangulations were further
investigated by Aronov et al.~\cite{ass-ctsp-93} who show that two
drawings admit compatible triangulations with only $O(n^2)$ new
vertices.  They moreover show that $\Theta(n^2)$ new vertices are
sometimes necessary.  This result has the following general
implication.  If there exist planar morphs between drawings of
triangulated graphs using $O(f(n))$ linear morphing steps, then there
are morphs between drawings of arbitrary plane graphs using
$O(f(n^2))$ steps.

To be able to morph with a polynomial number ($O(n^6)$) of linear
steps Lubiw and Petrick~\cite{lp-mpgdbe-08} relaxed the straight-line
requirement and showed how to morph between two planar drawings when
edges are allowed to be bent during the morph.  However, this result
can also be achieved without this relaxation.  Alamdari et
al.~\cite{aac-mpgdpns-13} show that for every pair of planar
straight-line drawings of a triangulated graph with an equally drawn
outer face there exists a planar morph consisting of a sequence of
$O(n^2)$ linear morphs.  This is the first result showing that a
polynomial number of morphing steps is sufficient.  Using the results
on compatible triangulations mentioned above~\cite{ass-ctsp-93} this
yields a morph with $O(n^4)$ linear steps for general plane graphs.

Floater and Gotsman~\cite{fg-hmti-99} introduced a completely
different approach to planar morphing of triangulations.  They make
use of the fact that in a planar drawing the position of each vertex
is a convex combination of the neighboring vertices and that
conversely fixing the coefficients of the convex combinations and
fixing the outer face yields a planar drawing.  This was shown by
Floater~\cite{f-psast-97} extending the results by
Tutte~\cite{t-crg-60,t-hdg-63}.  Floater and Gotsman~\cite{fg-hmti-99}
create a morph between two planar drawings by transforming the
coefficients of the corresponding convex combinations into one
another, yielding a sequence of coefficients and thus a sequence of
planar drawings.  Surazhsky and Gotsman~\cite{sg-cmcpt-01,sg-imct-03}
improve this approach further to obtain aesthetically more appealing
morphs.

The approach based on convex combinations has the disadvantage that
the trajectories are not explicitly computed and that it is not clear
how many linear morphing steps are necessary to obtain a planar and
smooth morph.  Despite its theoretical shortcomings, in practice this
algorithm leads to nice morphs, as shown by Erten et
al.~\cite{ekp-ifmpg-04,ekp-mpg-04}, who combine this approach with
rigid motion (translation, rotation, scaling and shearing) and the
triangulation algorithm by Aronov et al.~\cite{ass-ctsp-93}.
Moreover, they are able to morph edges with bends to straight-line
edges and vice versa.

Biedl et al.~\cite{bls-mpgped-06} consider a related problem of
morphing so called \emph{parallel} straight-line drawings, that is
straight-line drawings such that for every edge $e$, the slope of $e$
is the same in both drawings.  Moreover, the edge slopes have to be
preserved throughout the whole morph.  They show that for orthogonal
drawings (without bends) such a morph always exists.  On the other
hand, testing for the existence of such a morph becomes NP-hard if the
edges are allowed to have three or more slopes.  Lubiw et
al.~\cite{lps-mopgd-06} investigate morphs between general orthogonal
drawings of planar graphs where edges may have bends.  They show that
for every pair of drawings there is a morph preserving planarity and
orthogonality consisting of polynomial many steps, where each step is
either a movement of vertices or a ``twist'' around a vertex that
introduces new bends at the edges incident to this vertex.

Of course, problems similar to planar morphing can be considered for
non-planar graphs.  Examples are the results by Friedrich and
Eades~\cite{fe-gdm-02} and Friedrich and Houle~\cite{cm-gdm-02}.

\section{Open Problems}
\label{sec:open-problems}

There are many interesting problems, some of which have been open for
a decade and have resisted efforts to address them. Here we list
several of the current open problems.
\begin{enumerate}
\item Given two arbitrary planar graphs $\1G=(\1V,\1E)$ and
  $\2G=(\2V,\2E)$ with the same number of vertices, $|\1V|=|\2V|$,
  does there always exist a mapping from the vertex set of the first
  graph onto the vertex set of the second graph $\1V\rightarrow \2V$
  such that the two graphs have a {\sc SGE}? That is, do pairs of
  planar graphs always have an {\sc SGE} without mapping?
\item Given two graphs of max-degree 2, $\1G=(\1V,\1E)$ and
  $\2G=(\2V,\2E)$ with the same number of vertices, an {\sc SGE} with
  mapping does always exist.  Unlike most other results where the pair
  of graphs has an {\sc SGE} the area of the necessary grid is not
  bounded. Is it possible to guarantee polynomial integer grid for the
  simultaneous embedding?
\item What is the complexity of {\sc SGE} for two graphs with fixed
  planar embeddings?
\item Is it possible to decide {\sc SGE} for restricted cases, for
  example if the common graph is highly connected?
\item What is the complexity of the decision problem {\sc SEFE} for
  two graphs?
\item Are there interesting parameters for which {\sc SEFE} or {\sc
    SGE} are FPT?  For example, tree-distance of~$G$?  What about
  maximum degree~$\Delta$?
\item What is the complexity of {\sc SEFE} for more than two graphs
  with sunflower intersection?
\item What is the complexity of {\sc SEFE} for four graphs, each with
  a fixed planar embedding?
\item What is the complexity of the optimization version of {\sc SEFE}
  where one asks for drawings such that as many common edges as
  possible are drawn the same?
\item Let $\1G$ and $\2G$ be two planar graphs with given
  combinatorial embeddings inducing the same embedding on their
  intersection $G$, that is a {\sc SEFE} is given with the input.
  What is the complexity of minimizing the number of crossings in a
  corresponding drawing?
\item Let $\1G$ and $\2G$ be two planar graphs with given
  combinatorial embeddings inducing the same embedding on their
  intersection $G$, that is a {\sc SEFE} is given with the input.
  Do $\1G$ and $\2G$ admit drawings with few bends on a small grid
  respecting the given {\sc SEFE}?
\item There are many open problems in the {\sc CSE} setting. A
  particularly interesting one concerns pairs of trees. It is known
  that two $n$-vertex trees without mapping (1-colored) have a
  simultaneous geometric embedding (any set of $n$ points in convex
  position suffices). It is also known that at the other extreme when
  the mapping is given ($n$-colored) such geometric embedding may not
  exist. However, the problem is open for any number of colors $c\in
  \{2, \ldots, n-1\}$.
\item Similarly to the previous problem, the status of the tree-path {\sc CSE} problem is open for any number of colors $c\in \{3, \ldots, n-1\}$.
\end{enumerate}

\bibliographystyle{abbrv}
{
\bibliography{se}

\begin{thebibliography}{10}

\bibitem{aac-mpgdpns-13}
S.~Alamdari, P.~Angelini, T.~M. Chan, G.~{Di Battista}, F.~Frati, A.~Lubiw,
  M.~Patrignani, V.~Roselli, S.~Singla, and B.~T. Wilkinson.
\newblock Morphing planar graph drawings with a polynomial number of steps.
\newblock In {\em Proceedings of the twenty-fourth annual ACM-SIAM symposium on
  Discrete algorithm}, SODA '13. ACM, 2013.
\newblock To appear.

\bibitem{adfpr-tsegi-12}
P.~Angelini, G.~D. Battista, F.~Frati, M.~Patrignani, and I.~Rutter.
\newblock Testing the simultaneous embeddability of two graphs whose
  intersection is a biconnected or a connected graph.
\newblock {\em Journal of Discrete Algorithms}, 14(0):150--172, 2012.

\bibitem{fab-embedded-11}
P.~Angelini, G.~{Di Battista}, and F.~Frati.
\newblock Simultaneous embedding of embedded planar graphs.
\newblock In T.~Asano, S.-i. Nakano, Y.~Okamoto, and O.~Watanabe, editors, {\em
  Algorithms and Computation}, volume 7074 of {\em Lecture Notes in Computer
  Science}, pages 271--280. Springer Berlin / Heidelberg, 2011.

\bibitem{abf-10}
P.~Angelini, G.~{Di Battista}, F.~Frati, V.~Jel\'{\i}nek, J.~Kratochv\'{\i}l,
  M.~Patrignani, and I.~Rutter.
\newblock Testing planarity of partially embedded graphs.
\newblock In {\em Proceedings of the Twenty-First Annual ACM-SIAM Symposium on
  Discrete Algorithms}, SODA '10, pages 202--221. Society for Industrial and
  Applied Mathematics, 2010.

\bibitem{kaufmann-tp-12}
P.~Angelini, M.~Geyer, M.~Kaufmann, and D.~Neuwirth.
\newblock On a tree and a path with no geometric simultaneous embedding.
\newblock {\em Journal of Graph Algorithms and Applications}, 16(1):37--83,
  2012.

\bibitem{ap-mmphuodg-13}
D.~Archambault and H.~Purchase.
\newblock Mental map preservation helps user orientation in dynamic graphs.
\newblock In {\em Graph Drawing}, Lecture Notes in Computer Science. Springer
  Berlin / Heidelberg, 2013.
\newblock To appear.

\bibitem{abks-grsdg-12}
E.~Argyriou, M.~Bekos, M.~Kaufmann, and A.~Symvonis.
\newblock Geometric {RAC} simultaneous drawings of graphs.
\newblock In J.~Gudmundsson, J.~Mestre, and T.~Viglas, editors, {\em Computing
  and Combinatorics}, volume 7434 of {\em Lecture Notes in Computer Science},
  pages 287--298. Springer Berlin Heidelberg, 2012.

\bibitem{ass-ctsp-93}
B.~Aronov, R.~Seidel, and D.~Souvaine.
\newblock On compatible triangulations of simple polygons.
\newblock {\em Computational Geometry}, 3(1):27---35, 1993.

\bibitem{bk-btg-79}
F.~Bernhart and P.~C. Kainen.
\newblock The book thickness of a graph.
\newblock {\em Journal of Combinatorial Theory, Series B}, 27(3):320--331,
  1979.

\bibitem{bls-mpgped-06}
T.~Biedl, A.~Lubiw, and M.~Spriggs.
\newblock Morphing planar graphs while preserving edge directions.
\newblock In P.~Healy and N.~Nikolov, editors, {\em Graph Drawing}, volume 3843
  of {\em Lecture Notes in Computer Science}, pages 13--24. Springer Berlin /
  Heidelberg, 2006.

\bibitem{br-drpse-13}
T.~Bl{\"a}sius and I.~Rutter.
\newblock Disconnectivity and relative positions in simultaneous embeddings.
\newblock In {\em Graph Drawing}, Lecture Notes in Computer Science. Springer
  Berlin / Heidelberg, 2013.
\newblock To appear.

\bibitem{br-spoacep-13}
T.~Bl{\"a}sius and I.~Rutter.
\newblock Simultaneous {PQ}-ordering with applications to constrained embedding
  problems.
\newblock In {\em Proceedings of the twenty-fourth annual ACM-SIAM symposium on
  Discrete algorithm}, SODA '13. ACM, 2013.
\newblock To appear.

\bibitem{colored-se-11}
U.~Brandes, C.~Erten, A.~Estrella-Balderrama, J.~J. Fowler, F.~Frati, M.~Geyer,
  C.~Gutwenger, S.-H. Hong, M.~Kaufmann, S.~G. Kobourov, G.~Liotta, P.~Mutzel,
  and A.~Symvonis.
\newblock Colored simultaneous geometric embeddings and universal pointsets.
\newblock {\em Algorithmica}, 60(3):569--592, 2011.

\bibitem{bw-bpdgl-97}
U.~Brandes and D.~Wagner.
\newblock A bayesian paradigm for dynamic graph layout.
\newblock In G.~D. Battista, editor, {\em Proceedings of the 5th International
  Symposium on Graph Drawing (GD'97)}, volume 1353 of {\em Lecture Notes in
  Computer Science}, pages 236--247. Springer, 1997.

\bibitem{se-original-07}
P.~Brass, E.~Cenek, C.~A. Duncan, A.~Efrat, C.~Erten, D.~Ismailescu, S.~G.
  Kobourov, A.~Lubiw, and J.~S.~B. Mitchell.
\newblock On simultaneous planar graph embeddings.
\newblock {\em Computational Geometry: Theory and Applications},
  36(2):117--130, 2007.

\bibitem{bt-dmiog-98}
S.~Bridgeman and R.~Tamassia.
\newblock Difference metrics for interactive orthogonal graph drawing
  algorithms.
\newblock In S.~H. Whitesides, editor, {\em Proceedings of the 6th
  International Symposium on Graph Drawing (GD'98)}, volume 1547 of {\em
  Lecture Notes in Computer Science}, pages 51--71. Springer, 1998.

\bibitem{bfgtv-igaio-97}
S.~S. Bridgeman, J.~Fanto, A.~Garg, R.~Tamassia, and L.~Vismara.
\newblock Interactivegiotto: An algorithm for interactive orthogonal graph
  drawing.
\newblock In G.~D. Battista, editor, {\em Proceedings of the 5th International
  Symposium on Graph Drawing (GD'97)}, volume 1353 of {\em Lecture Notes in
  Computer Science}, pages 303--308. Springer, 1997.

\bibitem{bs-rpg-93}
G.~R. Brightwell and E.~R. Scheinerman.
\newblock Representations of planar graphs.
\newblock {\em SIAM Journal on Discrete Mathematics}, 6(2):214--229, 1993.

\bibitem{beppe-matched-11}
S.~Cabello, M.~J. van Kreveld, G.~Liotta, H.~Meijer, B.~Speckmann, and
  K.~Verbeek.
\newblock Geometric simultaneous embeddings of a graph and a matching.
\newblock {\em Journal of Graph Algorithms and Applications}, 15(1):79--96,
  2011.

\bibitem{c-dprc-44}
S.~S. Cairns.
\newblock Deformations of plane rectilinear complexes.
\newblock {\em The American Mathematical Monthly}, 51(5):247--252, 1944.

\bibitem{se-arcs-09}
J.~Cappos, A.~Estrella-Balderrama, J.~J. Fowler, and S.~G. Kobourov.
\newblock Simultaneous graph embedding with bends and circular arcs.
\newblock {\em Computational Geometry: Theory and Applications},
  42(2):173--182, 2009.

\bibitem{cn-hcp-89}
N.~Chiba and T.~Nishizeki.
\newblock The hamiltonian cycle problem is linear-time solvable for 4-connected
  planar graphs.
\newblock {\em Journal of Algorithms}, 10(2):187--211, 1989.

\bibitem{juenger-crossings-08}
M.~Chimani, M.~J{\"u}nger, and M.~Schulz.
\newblock Crossing minimization meets simultaneous drawing.
\newblock In {\em IEEE Pacific Visualisation Symposium}, pages 33--40, 2008.

\bibitem{h-uwukr-34}
C.~{Chojnacki (Haim Hanani)}.
\newblock {\"U}ber wesentlich unpl{\"a}ttbare kurven im dreidimensionalen
  raume.
\newblock {\em Fundamenta Mathematicae}, 23:135--142, 1934.

\bibitem{clr-egb-87}
F.~R.~K. Chung, F.~T. Leighton, and A.~L. Rosenberg.
\newblock Embedding graphs in books: a layout problem with applications to
  {VLSI} design.
\newblock {\em SIAM Journal on Algebraic and Discrete Methods}, 8(1):33--58,
  1987.

\bibitem{cknpw-sgbve-03}
C.~Collberg, S.~Kobourov, J.~Nagra, J.~Pitts, and K.~Wampler.
\newblock A system for graph-based visualizations of the evolution of software.
\newblock In {\em Proccedings of the Symposium on Visualization}, pages 77--86,
  212--213. ACM, 2003.

\bibitem{fpp-hdpgg-90}
H.~de~Fraysseix, J.~Pach, and R.~Pollack.
\newblock How to draw a planar graph on a grid.
\newblock {\em Combinatorica}, 10(1):41--51, 1990.

\bibitem{dett-gdavg-99}
G.~{Di Battista}, P.~Eades, R.~Tamassia, and I.~G. Tollis.
\newblock {\em Graph Drawing: Algorithms for the Visualization of Graphs}.
\newblock Prentice Hall, 1999.

\bibitem{ddlw-bespd-06}
E.~{Di Giacomo}, W.~Didimo, G.~Liotta, and S.~K. Wismath.
\newblock Book embeddability of series-parallel digraphs.
\newblock {\em Algorithmica}, 45(4):531--547, 2006.

\bibitem{ddkls-mdpg-09}
E.~{Di Giacomo}, W.~Didimo, M.~van Kreveld, G.~Liotta, and B.~Speckmann.
\newblock Matched drawings of planar graphs.
\newblock {\em Journal of Graph Algorithms and Applications}, 13(3):423--445,
  2009.

\bibitem{beppe-note-05}
E.~{Di Giacomo} and G.~Liotta.
\newblock A note on simultaneous embedding of planar graphs.
\newblock In {\em EuroCG}, pages 207--210, 2005.

\bibitem{beppe-outer-07}
E.~{Di Giacomo} and G.~Liotta.
\newblock Simultaneous embedding of outerplanar graphs, paths, and cycles.
\newblock {\em International Journal of Computational Geometry and
  Applications}, 17(2):139--160, 2007.

\bibitem{dg-gtac-02}
S.~Diel and C.~G\"org.
\newblock Graphs, they are changing -- dynamic graph drawing for a sequence of
  graphs.
\newblock In M.~T. Goodrich and S.~G. Kobourov, editors, {\em Proceedings of
  the 10th International Symposium on Graph Drawing (GD'02)}, volume 2528 of
  {\em Lecture Notes in Computer Science}, pages 23--31. Springer, 2002.

\bibitem{deh-gtcg}
M.~B. Dillencourt, D.~Eppstein, and D.~S. Hirschberg.
\newblock Geometric thickness of complete graphs.
\newblock {\em Journal of Graph Algorithms and Applications}, 4(3):5--17, 2000.

\bibitem{dek-gtldg-04}
C.~A. Duncan, D.~Eppstein, and S.~G. Kobourov.
\newblock The geometric thickness of low degree graphs.
\newblock In {\em Proceedings of the 20th annual symposium on Computational
  geometry}, SCG '04, pages 340--346. ACM, 2004.

\bibitem{ek-sepgfb-05}
C.~Erten and S.~Kobourov.
\newblock Simultaneous embedding of planar graphs with few bends.
\newblock In J.~Pach, editor, {\em Graph Drawing}, volume 3383 of {\em Lecture
  Notes in Computer Science}, pages 195--205. Springer Berlin / Heidelberg,
  2005.

\bibitem{ekp-ifmpg-04}
C.~Erten, S.~Kobourov, and C.~Pitta.
\newblock Intersection-free morphing of planar graphs.
\newblock In G.~Liotta, editor, {\em Graph Drawing}, volume 2912 of {\em
  Lecture Notes in Computer Science}, pages 320--331. Springer Berlin /
  Heidelberg, 2004.

\bibitem{se-dual-05}
C.~Erten and S.~G. Kobourov.
\newblock Simultaneous embedding of a planar graph and its dual on the grid.
\newblock {\em Theory of Computing Systems}, 38(3):313--327, 2005.

\bibitem{se-schemes-05}
C.~Erten, S.~G. Kobourov, V.~Le, and A.~Navabi.
\newblock Simultaneous graph drawing: Layout algorithms and visualization
  schemes.
\newblock {\em Journal of Graph Algorithms and Applications}, 9(1):165--182,
  2005.

\bibitem{ekp-mpg-04}
C.~Erten, S.~G. Kobourov, and C.~Pitta.
\newblock Morphing planar graphs.
\newblock In {\em Proceedings of the twentieth annual symposium on
  Computational geometry}, SCG '04, pages 451--452. ACM, 2004.

\bibitem{graphset-10}
A.~Estrella-Balderrama, J.~J. Fowler, and S.~G. Kobourov.
\newblock {GraphSET}, a tool for simultaneous graph drawing.
\newblock {\em Software: Practice and Experience}, 40(10):849--863, 2010.

\bibitem{juenger-np-07}
A.~Estrella-Balderrama, E.~Gassner, M.~J{\"u}nger, M.~Percan, M.~Schaefer, and
  M.~Schulz.
\newblock Simultaneous geometric graph embeddings.
\newblock In {\em Graph Drawing}, volume 4875 of {\em Lecture Notes in Computer
  Science}, pages 280--290. Springer Berlin / Heidelberg, 2008.

\bibitem{f-slrpg-48}
I.~F{\'a}ry.
\newblock On straight lines representation of planar graphs.
\newblock {\em Acta Scientiarum Mathematicarum}, 11:229--233, 1948.

\bibitem{f-psast-97}
M.~S. Floater.
\newblock Parametrization and smooth approximation of surface triangulations.
\newblock {\em Computer Aided Geometric Design}, 14:231--250, 1997.

\bibitem{fg-hmti-99}
M.~S. Floater and C.~Gotsman.
\newblock How to morph tilings injectively.
\newblock {\em Journal of Computational and Applied Mathematics},
  101(1–2):117--129, 1999.

\bibitem{fk-culpg-08}
J.~Fowler and S.~Kobourov.
\newblock Characterization of unlabeled level planar graphs.
\newblock In S.-H. Hong, T.~Nishizeki, and W.~Quan, editors, {\em Graph
  Drawing}, volume 4875 of {\em Lecture Notes in Computer Science}, pages
  37--49. Springer Berlin / Heidelberg, 2008.

\bibitem{juenger-sefespqr-09}
J.~J. Fowler, C.~Gutwenger, M.~J{\"u}nger, P.~Mutzel, and M.~Schulz.
\newblock An {SPQR}-tree approach to decide special cases of simultaneous
  embedding with fixed edges.
\newblock In {\em Graph Drawing}, volume 5417 of {\em Lecture Notes in Computer
  Science}, pages 157--168. Springer Berlin / Heidelberg, 2009.

\bibitem{sefe-11}
J.~J. Fowler, M.~J{\"u}nger, S.~G. Kobourov, and M.~Schulz.
\newblock Characterizations of restricted pairs of planar graphs allowing
  simultaneous embedding with fixed edges.
\newblock {\em Computational Geometry: Theory and Applications},
  44(8):385--398, 2011.

\bibitem{fab-sefe-06}
F.~Frati.
\newblock Embedding graphs simultaneously with fixed edges.
\newblock In {\em Graph Drawing}, volume 4372 of {\em Lecture Notes in Computer
  Science}, pages 108--113. Springer Berlin / Heidelberg, 2007.

\bibitem{se-constrained-09}
F.~Frati, M.~Kaufmann, and S.~G. Kobourov.
\newblock Constrained simultaneous and near-simultaneous embeddings.
\newblock {\em Journal of Graph Algorithms and Applications}, 13(3):447--465,
  2009.

\bibitem{fe-gdm-02}
C.~Friedrich and P.~Eades.
\newblock Graph drawing in motion.
\newblock {\em Journal of Graph Algorithms and Applications}, 6(3):353--370,
  2002.

\bibitem{cm-gdm-02}
C.~Friedrich and M.~Houle.
\newblock Graph drawing in motion ii.
\newblock In P.~Mutzel, M.~Jünger, and S.~Leipert, editors, {\em Graph
  Drawing}, volume 2265 of {\em Lecture Notes in Computer Science}, pages
  122--125. Springer Berlin / Heidelberg, 2002.

\bibitem{ft-odgd-08}
Y.~Frishman and A.~Tal.
\newblock Onlyne dynamic graph drawing.
\newblock {\em IEEE Transactions on Visualizations and Computer Graphics},
  14(4):727--740, 2008.

\bibitem{juenger-sefe-06}
E.~Gassner, M.~J{\"u}nger, M.~Percan, M.~Schaefer, and M.~Schulz.
\newblock Simultaneous graph embeddings with fixed edges.
\newblock In F.~Fomin, editor, {\em Graph-Theoretic Concepts in Computer
  Science}, volume 4271 of {\em Lecture Notes in Computer Science}, pages
  325--335. Springer Berlin / Heidelberg, 2006.

\bibitem{g-igstcg-74}
F.~Gavril.
\newblock The intersection graphs of subtrees in trees are exactly the chordal
  graphs.
\newblock {\em Journal of Combinatorial Theory, Series B}, 16(1):47--56, 1974.

\bibitem{kaufmann-trees-09}
M.~Geyer, M.~Kaufmann, and I.~Vrto.
\newblock Two trees which are self-intersecting when drawn simultaneously.
\newblock {\em Discrete Mathematics}, 309(7):1909--1916, 2009.

\bibitem{gru-cgig-83}
M.~C. Golumbic, D.~Rotem, and J.~Urrutia.
\newblock Comparability graphs and intersection graphs.
\newblock {\em Discrete Mathematics}, 43(1):37--46, 1983.

\bibitem{gbpd-dgdso-04}
C.~G\"org, P.~Birke, M.~Pohl, and S.~Diel.
\newblock Dynamic graph drawing of sequences of orthogonal and hierarchical
  graphs.
\newblock In J.~Pach, editor, {\em Proceedings of the 12th International
  Symposium on Graph Drawing (GD'04)}, volume 3383 of {\em Lecture Notes in
  Computer Science}, pages 228--238. Springer, 2004.

\bibitem{ghlmw-mdgpgt-09}
L.~Grilli, S.-H. Hong, G.~Liotta, H.~Meijer, and S.~Wismath.
\newblock Matched drawability of graph pairs and of graph triples.
\newblock In {\em WALCOM: Algorithms and Computation}, volume 5431 of {\em
  Lecture Notes in Computer Science}, pages 322--333. Springer Berlin /
  Heidelberg, 2009.

\bibitem{lubiw-testing-10}
B.~Haeupler, K.~Jampani, and A.~Lubiw.
\newblock Testing simultaneous planarity when the common graph is 2-connected.
\newblock In O.~Cheong, K.-Y. Chwa, and K.~Park, editors, {\em Algorithms and
  Computation}, volume 6507 of {\em Lecture Notes in Computer Science}, pages
  410--421. Springer Berlin / Heidelberg, 2010.

\bibitem{ht-papqt-08}
B.~Haeupler and R.~E. Tarjan.
\newblock Planarity algorithms via {PQ}-trees (extended abstract).
\newblock {\em Electronic Notes in Discrete Mathematics}, 31:143--149, 2008.
\newblock The International Conference on Topological and Geometric Graph
  Theory.

\bibitem{hn-tpbecgp-09}
S.-H. Hong and H.~Nagamochi.
\newblock Two-page book embedding and clustered graph planarity.
\newblock Technical Report 2009-004, Department of Applied Mathematics \&
  Physics, Kyoto University, 2009.

\bibitem{ht-ept-74}
J.~Hopcroft and R.~E. Tarjan.
\newblock Efficient planarity testing.
\newblock {\em Journal of the ACM}, 21(4):549--568, 1974.

\bibitem{lubiw-chordal-09}
K.~Jampani and A.~Lubiw.
\newblock The simultaneous representation problem for chordal, comparability
  and permutation graphs.
\newblock In F.~Dehne, M.~Gavrilova, J.-R. Sack, and C.~Tóth, editors, {\em
  Algorithms and Data Structures}, volume 5664 of {\em Lecture Notes in
  Computer Science}, pages 387--398. Springer Berlin / Heidelberg, 2009.

\bibitem{lubiw-interval-10}
K.~Jampani and A.~Lubiw.
\newblock Simultaneous interval graphs.
\newblock In O.~Cheong, K.-Y. Chwa, and K.~Park, editors, {\em Algorithms and
  Computation}, volume 6506 of {\em Lecture Notes in Computer Science}, pages
  206--217. Springer Berlin / Heidelberg, 2010.

\bibitem{jkr-kttppeg-11}
V.~Jel\'{\i}nek, J.~Kratochv\'{\i}l, and I.~Rutter.
\newblock A {K}uratowski-type theorem for planarity of partially embedded
  graphs.
\newblock In {\em Proceedings of the 27th Annual ACM Symposium on Computational
  Geometry (SoCG'11)}, pages 107--116. ACM, 2011.

\bibitem{juenger-sefeinter-09}
M.~J{\"u}nger and M.~Schulz.
\newblock Intersection graphs in simultaneous embedding with fixed edges.
\newblock {\em Journal of Graph Algorithms and Applications}, 13(2):205--218,
  2009.

\bibitem{kammer-06}
F.~Kammer.
\newblock Simultaneous embedding with two bends per edge in polynomial area.
\newblock In L.~Arge and R.~Freivalds, editors, {\em Algorithm Theory – SWAT
  2006}, volume 4059 of {\em Lecture Notes in Computer Science}, pages
  255--267. Springer Berlin / Heidelberg, 2006.

\bibitem{kw-evpfb-02}
M.~Kaufmann and R.~Wiese.
\newblock Embedding vertices at points: Few bends suffice for planar graphs.
\newblock {\em Journal of Graph Algorithms and Applications}, 6(1):115--129,
  2002.

\bibitem{kkv-eprig-11}
P.~Klav\'{\i}k, J.~Kratochv\'{\i}l, and T.~Vysko\v{c}il.
\newblock Extending partial representations of interval graphs.
\newblock In {\em Proceedings of the 8th annual conference on Theory and
  applications of models of computation}, TAMC'11, pages 276--285.
  Springer-Verlag, 2011.

\bibitem{se-interactive-04}
S.~G. Kobourov and C.~Pitta.
\newblock An interactive multi-user system for simultaneous graph drawing.
\newblock In {\em Graph Drawing}, volume 3383 of {\em Lecture Notes in Computer
  Science}, pages 492--501. Springer Berlin / Heidelberg, 2005.

\bibitem{lp-mpgdbe-08}
A.~Lubiw and M.~Petrick.
\newblock Morphing planar graph drawings with bent edges.
\newblock {\em Electronic Notes in Discrete Mathematics}, 31(0):45--48, 2008.

\bibitem{lps-mopgd-06}
A.~Lubiw, M.~Petrick, and M.~Spriggs.
\newblock Morphing orthogonal planar graph drawings.
\newblock In {\em Proceedings of the seventeenth annual ACM-SIAM symposium on
  Discrete algorithm}, SODA '06, pages 222--230. ACM, 2006.

\bibitem{mels-lamm-95}
K.~Misue, P.~Eades, W.~Lai, and K.~Sugiyama.
\newblock Layout adjustment and the mental map.
\newblock {\em Journal of Visual Languages and Computing}, 6:183--210, 1995.

\bibitem{m-cpmpt-97}
B.~Mohar.
\newblock Circle packings of maps in polynomial time.
\newblock {\em European Journal of Combinatorics}, 18(7):785--805, 1997.

\bibitem{mr-tvrmt-98}
B.~Mohar and P.~Rosenstiehl.
\newblock Tessellation and visibility representations of maps on the torus.
\newblock {\em Discrete \& Computational Geometry}, 19:249--263, 1998.

\bibitem{mutzelthickness98}
P.~Mutzel, T.~Odenthal, and M.~Scharbrodt.
\newblock The thickness of graphs: A survey.
\newblock {\em Graphs and Combinatorics}, 14:59--73, 1998.

\bibitem{nw-ohgd-02}
S.~C. North and G.~Woodhall.
\newblock Online hierarchical graph drawing.
\newblock In P.~Mutzel, M.~J\"unger, and S.~Leipert, editors, {\em Proceedings
  of the 9th International Symposium on Graph Drawing (GD'01)}, volume 2265 of
  {\em Lecture Notes in Computer Science}, pages 232--246. Springer, 2002.

\bibitem{pw-epgfvl-98}
J.~Pach and R.~Wenger.
\newblock Embedding planar graphs at fixed vertex locations.
\newblock In S.~Whitesides, editor, {\em Graph Drawing}, volume 1547 of {\em
  Lecture Notes in Computer Science}, pages 263--274. Springer Berlin /
  Heidelberg, 1998.

\bibitem{phg-himm-07}
H.~C. Purchase, E.~Hoggan, and C.~G\"org.
\newblock How important is the ``mental map''? -- an empirical investigation of
  a dynamic graph layout algorithm.
\newblock In M.~Kaufmann and D.~Wagner, editors, {\em Proceedings of the 14th
  International Symposium on Graph Drawing (GD'06)}, volume 4372 of {\em
  Lecture Notes in Computer Science}, pages 184--195. Springer, 2007.

\bibitem{ps-ear-08}
H.~C. Purchase and A.~Samra.
\newblock Extremes are better: Investigating mental map preservation in dynamic
  graphs.
\newblock In G.~Stapleton, J.~Howse, and J.~Lee, editors, {\em Proceeding of
  the 5h International Symposium on Diagrammatic Representation and Inference
  (DIAGRAMS'08)}, volume 5223 of {\em Lecture Notes in Artifical Intelligence},
  pages 60--73. Springer, 2008.

\bibitem{smm-cvnldg-13}
A.~Sallaberry, C.~Muelder, and K.-L. Ma.
\newblock Clustering, visualizing, and navigating for large dynamic graphs.
\newblock In {\em Graph Drawing}, Lecture Notes in Computer Science. Springer
  Berlin / Heidelberg, 2013.
\newblock To appear.

\bibitem{s-ttp-13}
M.~Schaefer.
\newblock Toward a theory of planarity: Hanani-tutte and planarity variants.
\newblock In {\em Graph Drawing}, Lecture Notes in Computer Science. Springer
  Berlin / Heidelberg, 2013.
\newblock To appear.

\bibitem{s-cm-51}
S.~K. Stein.
\newblock Convex maps.
\newblock {\em Proceedings of the American Mathematical Society},
  2(3):464--466, 1951.

\bibitem{stt-mvuhss-81}
K.~Sugiyama, S.~Tagawa, and M.~Toda.
\newblock Methods for visual understanding of hierarchical system structures.
\newblock {\em IEEE Transactions on Systems, Man and Cybernetics},
  11(2):109--125, 1981.

\bibitem{sg-cmcpt-01}
V.~Surazhsky and C.~Gotsman.
\newblock Controllable morphing of compatible planar triangulations.
\newblock {\em ACM Transactions on Graphics}, 20(4):203--231, 2001.

\bibitem{sg-imct-03}
V.~Surazhsky and C.~Gotsman.
\newblock Intrinsic morphing of compatible triangulations.
\newblock {\em International Journal of Shape Modeling}, 9(2):191--201, 2003.

\bibitem{tt-trpg-89}
R.~Tamassia and I.~G. Tollis.
\newblock Tessellation representations of planar graphs.
\newblock In {\em Proceedings of the 27th Annual Allerton Conference on
  Communication, Control, and Computing}, pages 48--57, 1989.

\bibitem{t-dpg-83}
C.~Thomassen.
\newblock Deformations of plane graphs.
\newblock {\em Journal of Combinatorial Theory, Series B}, 34(3):244--257,
  1983.

\bibitem{t-ttcn-70}
W.~Tutte.
\newblock Toward a theory of crossing numbers.
\newblock {\em Journal of Combinatorial Theory}, 8(1):45--53, 1970.

\bibitem{t-crg-60}
W.~T. Tutte.
\newblock Convex representations of graphs.
\newblock {\em Proceedings of the London Mathematical Society}, 10:304--320,
  1960.

\bibitem{t-hdg-63}
W.~T. Tutte.
\newblock How to draw a graph.
\newblock {\em Proceedings of the London Mathematical Society}, 13:743--768,
  1963.

\bibitem{w-bzv-36}
K.~Wagner.
\newblock Bemerkungen zum {V}ierfarbenproblem.
\newblock {\em Jahresbericht der Deutschen Mathematiker-Vereinigung},
  46:26--32, 1936.

\bibitem{JCSS::Yannakakis1989}
M.~Yannakakis.
\newblock Embedding planar graphs in four pages.
\newblock {\em Journal of Computer and System Sciences}, 38(1):36--67, 1989.

\bibitem{zh-sslge-06}
H.~Zhang and X.~He.
\newblock On simultaneous straight-line grid embedding of a planar graph and
  its dual.
\newblock {\em Information Processing Letters}, 99(1):1--6, 2006.

\end{thebibliography}
}

\end{document}